\documentclass[11pt,A4,twoside]{article}  
\usepackage{geometry}
\usepackage{natbib}\def\citeN{\citet*}
\usepackage{dsfont}

\usepackage{graphicx}
\usepackage{booktabs} 
\usepackage{array} 
\usepackage{paralist} 
\usepackage{verbatim} 
\usepackage{subfig} 
\usepackage{amsmath}
\usepackage{amssymb}
\usepackage{amsthm}
\usepackage{todonotes}
\usepackage{indentfirst}
\usepackage{bbm}
\usepackage{stmaryrd}
\usepackage[ruled,vlined]{algorithm2e}
\usepackage{commath}
\usepackage{footnote}
\makesavenoteenv{algorithm}
\usepackage{rotating}
\usepackage{caption}
\usepackage{url}
\usepackage{multirow}

\graphicspath{{./figs/}}

\usepackage{fancyhdr} 
\usepackage{sectsty}

\usepackage[nottoc,notlof,notlot]{tocbibind} 
\usepackage[titles,subfigure]{tocloft} 

\def\D{\mathrm{d}}

\DeclareMathOperator*{\argmin}{argmin}

\newtheorem{heuristic}{Heuristic}

\usepackage{filecontents}

\def\theX{\widehat{X}}\def\theX{X} 
\def\theY{\widehat{Y}}\def\theY{Y}\def\theY{\psi}\def\theY{\xi_{i:n}}\def\theY{\xi}
\def\theI{\widehat{I}}\def\theI{I}\def\theI{Y} 
\def\theE{\widehat{E}}\def\theE{E}\def\theE{\Pi}

\def\thenb{N_{\text{steps}}^{\text{c}}} \def\thenb{n_c}\def\thenb{n_b}\def\thenb{n}
  
\def\theh{L} \def\theh{h}
\def\cB{\mathcal{B}}  
\def\coloneqq{=}
\def\eqqcolon{=}

\def\thef{\wideahat{f}}\def\thef{f}
\def\thee{e}

\def\lefta{} 
\def\righta{}
\def\supi{^i}\def\supi{}

\def\b{}\renewcommand\b[1]{{\color{blue}#1}}
\def\finproof {\rule{4pt}{6pt}}
\def\mx{N}
\def\my{M}

\def\tpi{\tp_i } \def\tpi{i}

\def\tjp{\tjp }\def\tjp{j+1}

\def\sj{\sj }\def\sj{j}
\def\si{\si }\def\si{i}
      
 \def\Var{\mathbb{V}\mathrm{ar}}
 
\newtheorem{theorem}{Theorem}
\newtheorem{lem}{Lemma}
\newtheorem{pro}{Proposition}
\newtheorem{cor}{Corollary}
\newtheorem{conj}{Conjecture}

\newtheorem{rem}{Remark}
\newtheorem{com}{Comments}
\newtheorem{ex}{Example}
\newtheorem{nota}{Notation}
\newtheorem{defi}{Definition}
\newtheorem{hyp}{Assumption}

\newcommand{\bt}{\begin{pro}}\newcommand{\et}{\end{pro}} 
\newcommand{\bl}{\begin{lem}}\newcommand{\el}{\end{lem}}
\newcommand{\bp}{\begin{pro}}\newcommand{\ep}{\end{pro}}
\newcommand{\bcor}{\begin{cor}}\newcommand{\ecor}{\end{cor}}
\newcommand{\bconj}{\begin{conj}}\newcommand{\econj}{\end{conj}}

\newcommand{\bd}{\begin{defi} \rm }\newcommand{\ed}{\finproof\end{defi} }
\newcommand{\eds}{\end{defi} }
\newcommand{\brem }{\begin{rem} 
}\newcommand{\erem }{
\end{rem}}
\newcommand{\bcom}{\begin{com} \rm }\newcommand{\ecom }{\end{com}}
\newcommand{\brems }{\begin{rem} 
}
\newcommand{\erems }{\end{rem}}
\newcommand{\bex}{\begin{ex} 
}\newcommand{\eex}{\finproof\end{ex}}
\newcommand{\bno}{\begin{nota} \rm }\newcommand{\eno}{\finproof\end{nota}}\newcommand{\enos}{\end{nota}}
\newcommand{\eexs}{\end{ex} }
\newcommand{\bhyp}{\begin{hyp} \rm }\newcommand{\ehyp}{\end{hyp}}
 
\newcommand{\bP}{\begin{proof}}
\newcommand{\eP}{\end{proof}}
\def\proof{\noindent {\it Proof. $\, $}}\def\proof{\noindent{\it\textbf{Proof}. $\, $}}
\def\finproof {\hfill $\Box$ \vskip 5 pt }\def\finproof{\rule{4pt}{6pt}}

 \newcommand{\beqa}{\begin{eqnarray}}
\newcommand{\eeqa}{\end{eqnarray}}
 \newcommand{\bea}{\begin{eqnarray*}}
\newcommand{\eea}{\end{eqnarray*}}
\def\bal{\begin{aligned}}
\def\eal{\end{aligned}}
\def\bll#1{
\beqa\label{#1}\bal
}\def\lel{\eal\eeqa}
\newcommand{\beql}[1]{\bll{#1}}
\newcommand{\eeql}{\lel}
\newcommand{\bel}{\bea\bal}
\newcommand{\eel}{\end{aligned}\eea}

\def\eee{\end{document}}

\newcommand{\op}[1]{#1}
\newcommand{\tmop}[1]{\ensuremath{\operatorname{#1}}}

\newcommand{\tmtextit}[1]{\emph{#1}}
\def\coloneqq{\coloneqq}
\def\coloneqq{:=}\def\coloneqq{=}

\def\sp{\,,\;}

\def\thestar{\thestar}\def\thestar{*}

\def\oversimulation{hierarchical simulation\xspace}

\def\Oversimulation{Hierarchical simulation\xspace}
\def\OverSimulation{Hierarchical Simulation\xspace}

\def\ie{\itshape i.e.}\def\ie{i.e.\xspace} 
 
\def\kappa{\b{P}}\def\kappa{P}
\def\Q{\mathbb{P}}\def\Q{\mathbb{Q}}
\def\Psi{B}

\def\ind{\mathds{1}}

\begin{document}
\title{Pathwise CVA Regressions With Oversimulated Defaults\footnote{A GPU implementation of our hierarchical simulation and regression learning scheme, as well as single-file python notebook demo for our CVA use case, 
are available   
on https://github.com/BouazzaSE/NeuralXVA.}} 

	\author{Lokman A. Abbas-Turki\textsuperscript{1},  St\'ephane Cr\'epey\textsuperscript{2} and Bouazza Saadeddine\textsuperscript{3,4}}

{\let\thefootnote\relax\footnotetext{\textsuperscript{1} \textit{Sorbonne Universit\'e, Laboratoire de Probabilit\'es et Mod\`eles Al\'eatoires, CNRS UMR 8001, Paris, France}}}

{\let\thefootnote\relax\footnotetext{\textsuperscript{2} \textit{Universit\'e Paris Cit\'e, Laboratoire de Probabilit\'es et Mod\`eles Al\'eatoires, CNRS UMR 8001, Paris, France  (\textbf{corresponding author}, {\tt stephane.crepey@lpsm.paris})}}}

{\let\thefootnote\relax\footnotetext{\textsuperscript{3} \textit{Universit\'e Paris-Saclay, Univ Evry, Laboratoire de Mathématiques et Modélisation d'Évry, CNRS UMR 8001, Evry, France}}} 

{\let\thefootnote\relax\footnotetext{\textsuperscript{4} 
\textit{Crédit Agricole CIB, Quantitative Research GMD/GMT, Paris, France.}}}


{\let\thefootnote\relax\footnotetext{{\it Acknowledgement:} This article has been accepted for publication in  Mathematical Finance, published by Wiley. We thank an anonymous referee for his useful comments. This research has benefited from the support of the Chair \textit{Capital Markets Tomorrow: Modeling and Computational Issues  under the aegis of the Institut Europlace de Finance},  a joint initiative of Laboratoire de Probabilités, Statistique et Modélisation (LPSM) / Université Paris Cité and  Crédit Agricole CIB. }}

\date{\today}

\maketitle

\rm
\begin{abstract} 
We consider the computation by simulation and neural net regression of conditional expectations, or more general elicitable statistics, of functionals of 
processes $(\theX,\theI)$. 
Here an exogenous component $\theI$ (Markov by itself) is time-consuming to simulate, while the endogenous component $\theX$ (jointly Markov with $\theI$) is  
quick to simulate given $\theI$, but is responsible for most of the variance of the simulated payoff.  
To address the related variance issue,
we introduce a conditionally independent, hierarchical simulation scheme, where several paths of $X$ are simulated for each simulated path of $Y$.
We analyze the statistical convergence of the regression learning scheme based on such block-dependent data.
We derive heuristics on the number of paths of $Y$ and, for each of them, of $X$, that should be simulated. 
The resulting algorithm is implemented on a graphics processing unit (GPU) combining Python/CUDA and learning with PyTorch. 
A CVA case study with a nested Monte Carlo benchmark shows that the \oversimulation technique is key to the success of the learning approach. 
\end{abstract}

\def\keywordname{{\bfseries Keywords:}}
\def\keywords#1{\par\addvspace\baselineskip\noindent\keywordname\enspace
\ignorespaces#1}\begin{keywords} 
\oversimulation,
neural net regression,
machine learning, 
X-valuation adjustment (XVA).
\end{keywords}
 
\vspace{2mm}
\noindent
\textbf{Mathematics Subject Classification:} 
91B25, 
91G40, 
62G08, 
62M45, 
68Q32. \\



\section{Introduction\label{s:intro}}

\subsection{Financial Motivation}
Greensill defaulted on March 8, 2021, a collapse estimated by  British  parliamentarians
to trigger a cost for UK taxpayers of up to £5bn\footnote{cf.~https://www.theguardian.com/business/2021/apr/28/greensill-collapse-could-cost-uk-taxpayer-up-to-5bn-mps-told, accessed on May 16, 2022.}.
Greensill fell short of capital because they lent to Gupta against future invoices which then did not materialize.  
A projection of  Greensill's capital requirements
including the
tail risk related to Gupta's default
would have highlighted a sizable concentrated and unsecured credit risk to a junk-rated counter-party. This example emphasizes the importance of 
performing proper
 default risk simulations, as opposed to credit spread simulations simply, as typically done in the industry for the sake of simplicity.
Another example (in fact, the motivation for this work) is the path-wise XVA regression approach introduced in \citeN{CrepeyHoskinsonSaadeddine2019}, 
which also requires
a hybrid market and credit setup, where the actual defaults of the clients of the bank are simulated.

However, to obtain the 
required level of accuracy in a simulation setup with defaults, one needs a very large number of simulations---100 times more, say, than in a purely diffusive model, where 100 is in reference to credit spreads that would be of the order of 1\%. A factor 100 is not necessarily a sizable amount as far as the simulation of the risk factors is involved. But it also means that the mark-to-market cube of path-wise prices of all trades of the bank becomes 100 times bigger. When applied at the level of a realistic banking portfolio, this becomes prohibitive in terms of both computation time and memory occupancy.
 
\subsection{Contribution of the Paper}
To overcome this problem, we introduce an acceleration technique for the computation by simulation and regression of conditional expectations of functions $\theY$ of Markov pairs $(\theX,\theI)$, where $\theI$ is an exogenous component,  Markov by itself, whose simulation is time-consuming, while the endogenous component $\theX$ (jointly Markov with $\theI$) is quick to simulate given $Y$, but also responsible for
most of the variance of the simulated payoff. The idea, which we call \oversimulation, 
is then to draw an optimized number of realizations of $X$ conditional on each simulation of $Y$. For example, in the above-mentioned 
 XVA regression framework, we simulate a few hundred paths of client defaults conditionally on each mark-to-market path. Proceeding in this way, 
the computational burden of the mark-to-market cube is not amplified by the simulation of the client defaults. We demonstrate, both mathematically and empirically, that the lack of independence of the ensuing simulation setup is not detrimental to the quality of the
ensuing learner,
i.e.~(in the above case) of the regressions of the XVA layers built over the mark-to-market cube and defaults scenarios. In addition,
an a posteriori twin Monte Carlo validation technique is introduced to estimate the $L_2$ error between a targeted conditional expectation and any estimator for the latter. 
 
\subsection{Related literature}

{Supervised learning tasks can be subdivided in two main categories \citep{Murphy12}. The first one is classification problems, for which training a neural net allows learning a (conditional) probability that a discrete random variable takes its different possible values (such as ‘cat or dog’ for an image), by empirical minimization of an entropic criterion \footnote{or Kullback-Leibler divergence, or maximization  of a likelihood.},  based on a number of labeled observations (images and their correct classification, e.g. `cat or dog'). The second category consists of regression tasks, for which  training a neural net allows learning a conditional expectation  $\mathbb{E}(V | U)$, by empirical minimization of a least squares criterion based on observations of the pair of random variables $(U,V)$. Our paper falls in the second category,  in the special case where the $(U,V)$ are simulated, assuming the data generating process known. For this purpose we rely, as is standard, on the $L_2$ training error and on out-of-sample validation, represented in our case by companion twin and nested Monte Carlo procedures.}

In the last years, such neural net regression-based simulation techniques have rapidly imposed themselves as a worthy player in the field of the numerical methods for PDEs and BSDEs: see e.g.~\citeN{EHanJentzen17} or \citeN{HurePhamWarin19} (to quote only two). 
There are good reasons for this evolution, starting with universal approximation theorems, or density results justifying the use of neural network as a versatile parameterization: see \citep{kidger2020universal} for the fixed-width case and \citep{hornik1991approximation, cybenko1989approximation} for the fixed-depth case. Other incentives for using neural networks are the ability to automatically infer a linear regression basis (provided by the trained hidden layers\footnote{see Figure \ref{fig:NN}.}), or the ease of transfer learning \citep{pan2009survey, bozinovski2020reminder}. However, the control of the error arising from the numerical optimizations for the embedded training tasks (a priori error bounds for \eqref{e:regrkappo} in our setup) is a largely open issue: {in the nonconvex numerical optimization case typical of neural net training tasks, only idealized versions of stochastic gradient descents with over-parameterized neural nets can be shown to converge to the global minimum of the empirical loss, and only in-sample \citep*{chizat2018global,du2019gradient}; Otherwise, only local minima can be guaranteed, under assumptions such as the ones in~\citeN{lei2019stochastic}. These local minima are often close to global minima \citep*{choromanska2015loss}, so that this local vs. global minimization issue is considered by many as a false problem in  practice.}
Still, our aforementioned twin Monte Carlo validation procedure is a useful practical contribution in this regard.

On the XVA side,
the use of regression-based Monte Carlo simulations is not a novelty  by itself. It was already presented in \citeN{cesari2010modelling}
 as a key CVA computational paradigm, intended to avoid nested Monte Carlo. However, from such traditional XVA computations to the neural net regressions of \citeN{HugeSavine20},
 the regressions are only used for computing the mark-to-market cube of the prices of all the contracts of the bank with all its clients (or netting sets) at all times of a simulation time-grid, out of which the CVA of the bank at time 0 (and only it) is obtained by integration of the so-called expected positive exposure relative to each netting set against the credit curve of the corresponding client, and summation over netting sets. By contrast, in this paper, we aim at learning the CVA as a process, i.e.~at every node of a simulation for all risk factors, based on a mark-to-market cube computed by model analytics at the forward simulation stage. Regressions could also be used here, but these would be more standard, hence we ignore them in this paper: regressions for the mark-to-market of derivatives à la \citeN{cesari2010modelling} 
  are typically
 multiple parametric regressions in diffusive and low-dimensional setups, as opposed to hybrid diffusive / Markov chain setup and high-dimensional neural net regressions in the path-wise CVA case targeted in this work. Recently, \citeN{gnoatto2020deep} 
 deep-hedge and learn the CVA and the FVA, but this is again in a purely diffusive setup, after the default of the bank and its (assumed single) counterparty have been eliminated from the model by the reduction of filtration technique of \citeN{Crepey13a1}. This technique of reduction of filtration is not extendible to the realistic case of a bank involved in transactions with several (in practice, many, e.g.~several thousands of) clients, the default times of which enter the ensuing FVA (and KVA) equations in a nonlinear fashion, so that there is then no other choice but simulating these defaults and including them in the training. But this requires special care, which is the topic of this work.  
 
\subsection{Outline}

In Section~\ref{s:setup}, we introduce a neural net
learning framework for 
conditional expectations, iterated in time as they appear naturally in dynamic pricing problems,
taking into account the 
dynamics of the problem by means of a
backward pricing algorithm.   
A twin Monte Carlo validation technique is also introduced to estimate the error of the estimator.
In Section~\ref{s:hierarch}, we identify a variance issue raised by the coexistence of risk factors evolving at different paces in the problem (e.g. market risk versus default indicator processes)
and we propose a \oversimulation approach to address it. 
We establish the benefit of this approach mathematically by providing associated generalization bounds. 
Section~\ref{sec:cva} illustrates our approach with a CVA numerical case study.

\brem\label{rem:elici} Although our CVA case study only covers quadratic risk minimization (for benchmarking reasons), the approach and the proofs of this paper are valid for more general loss functions and apply to the learning of any elicitable statistics.
In particular, via the \citeN{rock:urya:00} representation of value-at-risk and expected shortfall of a given loss (random variable) in terms of ``far out-of-the-money call options'' on that loss, our hierarchical simulation approach is also relevant for learning value-at-risk and expected shortfall 
in hybrid mark-to-market and default simulation setups.
 Such an approach is even particularly relevant in these cases, where the fact that $\theX$ is responsible for most of the variance of the payoff is then intrinsic to the far out-of-the-money  feature of the corresponding ``option''.
 \erem

\section{Neural 
Regression Setup}\label{s:setup}

A reference probability space, with corresponding probability measure and expectation denoted by $\Q$ 
and $\mathbb{E}$, is fixed throughout the paper. The state spaces of $\theX$ and $\theI$ are taken as  $\mathbb{R}^p$ and $\mathbb{R}^q $, for some positive integers $p$ and $q$. We identify
$\mathbb{R}^{p+q}$ with $\mathbb{R}^p\times\mathbb{R}^q$ and write $\phi(z)$ or $\phi(x, y)$  interchangeably, where $z$ is the concatenation of $x$ and $y$, for every function $\phi$ defined over
$\mathbb{R}^{p+q}$  or $\mathbb{R}^p\times\mathbb{R}^q$, and for every $x\in\mathbb{R}^p$ and $y\in\mathbb{R}^q$.

In the (default risk)
case of a Markov chain like component $X$, 
referred to hereafter as the Markov chain $X$ case (but with transition intensities modulated by $Y$), we assume, without loss of generality in this case,
that $X$ evolves on the vertices $\{0,1\}^p$ of the unit cube in $\mathbb{R}^p$.
We take the problem after discretisation of time (if the latter was continuous in the first place), for a time step set to one year for ease of notation.

We then consider 
$(X_i)_{0\leq i\leq n}$ and $(Y_i)_{0\leq i\leq n}$  
as discrete-time processes on the 
time grid. 
Our goal is to estimate, for every $\tpi  ,$
conditional expectations of the form 
\begin{equation}\label{cond_exp_discrete} \Pi_i = \mathbb{E}[\xi_{i,n} \supi \left|X_i, Y_i\right.],\end{equation}
where
\begin{equation}\label{cond_exp_discretebis} \xi_{i,n}  \coloneqq f_i
(X_i, \dots, X_n, Y_i, \dots, Y_n).
\end{equation}
Here $\thef_i$ is  a measurable real function
 such that $\xi_{i,n}$ is a square-integrable random variable. 


Conditional expectations such as \eqref{cond_exp_discrete} can be estimated via linear regression using a finite sample. This is ubiquitous in quantitative finance since the Bermudan Monte Carlo papers of \citeN{tsitsiklis2001regression} and \citeN{longstaff2001valuing}. 
In order to estimate the conditional expectation in \eqref{cond_exp_discrete}, one draws {i.i.d.} samples 
$\{(X_i^{\iota}, Y_i^{\iota}, \xi_{i,n}^{\iota})\}_{\iota\in\mathcal{I}}$ of $(X_i, Y_i, \xi_{i,n} )$, where $\mathcal{I}$ is a finite set of indices. Then, given a feature map $\phi:\mathbb{R}^{p+q}\rightarrow\mathbb{R}^m$ (for some positive integer $m$), 
one linearly regresses $\{\xi_{i,n}^\iota\}_{\iota\in\mathcal{I}}$ against $\{\phi(X_i^\iota, Y_i^\iota)\}_{\iota\in\mathcal{I}}$,~solving for 
\beql{e:forw}\hat{w}_i\in\argmin_{w_i \in\mathbb{R}^m} \sum_{\iota\in\mathcal{I}}(\xi_{i,n}^\iota - w_i^{\top}\phi(X_i^\iota, Y_i^\iota) )^2.
\eeql 
One then uses $\hat{w}^{\top}_i\phi(\theX_{\tpi }, \theI_{\tpi }) $ 
as an approximation for $\theE  _{\tpi }$.

The above procedure is justified by the characterization, in the square integrable case, of conditional expectations as orthogonal projections, i.e.~\[\mathbb{E}[\xi_{i,n} \left|\theX_{\tpi }, \theI_{\tpi }\right.]=\varphi^{\star}_i(\theX_{\tpi }, \theI_{\tpi })\quad\text{a.s.},\]
where, denoting by $\cB(E)$ the set of Borel measurable real functions on a
metric
space $E$, 
\begin{equation}\label{e:genregr}
 \varphi^{\star}_i\in\argmin_{\varphi_i\in\cB(\mathbb{R}^p {\times\mathbb{R}^q})} \mathbb{E}[(\xi_{i,n} -\varphi_i(\theX_{\tpi }, \theI_{\tpi }))^2].\end{equation}
One recovers the linear regression formulation \eqref{e:forw} by approximating the expectation by an empirical mean and restricting the search space to the functions of the form $\mathbb{R}^{p+q} \owns (x,y ) \mapsto
w_i^{\top}\phi(x, y)$, where $w_i\in\mathbb{R}^m$.

\subsection{Neural Net Parameterization}

Linear regression by means of a priori, explicit factors has a reasonable chance of success when $\varphi^{\star}_i$ is simple enough and the feature mapping $\phi$ can be judiciously chosen, usually from expert knowledge. This is however not always the case, e.g.~when considering portfolio-wide XVA metrics, which exhibit non-trivial dependencies on the many risk factors being regressed against. It is then impossible to manually devise a satisfactory feature mapping $\phi$.
Figure~\ref{fig:NNvsLinearCall} shows how a linear regression with the raw risk factors as features fails in learning the (conditional) CVA of an elementary portfolio made of a single call option, while the neural net estimator almost matches with the nested Monte Carlo estimator (see Section \ref{sec:cva} for more numerical details).
\begin{figure}
\begin{minipage}{.49\textwidth}
 \center
\vspace{-3mm} \includegraphics[width=\textwidth]{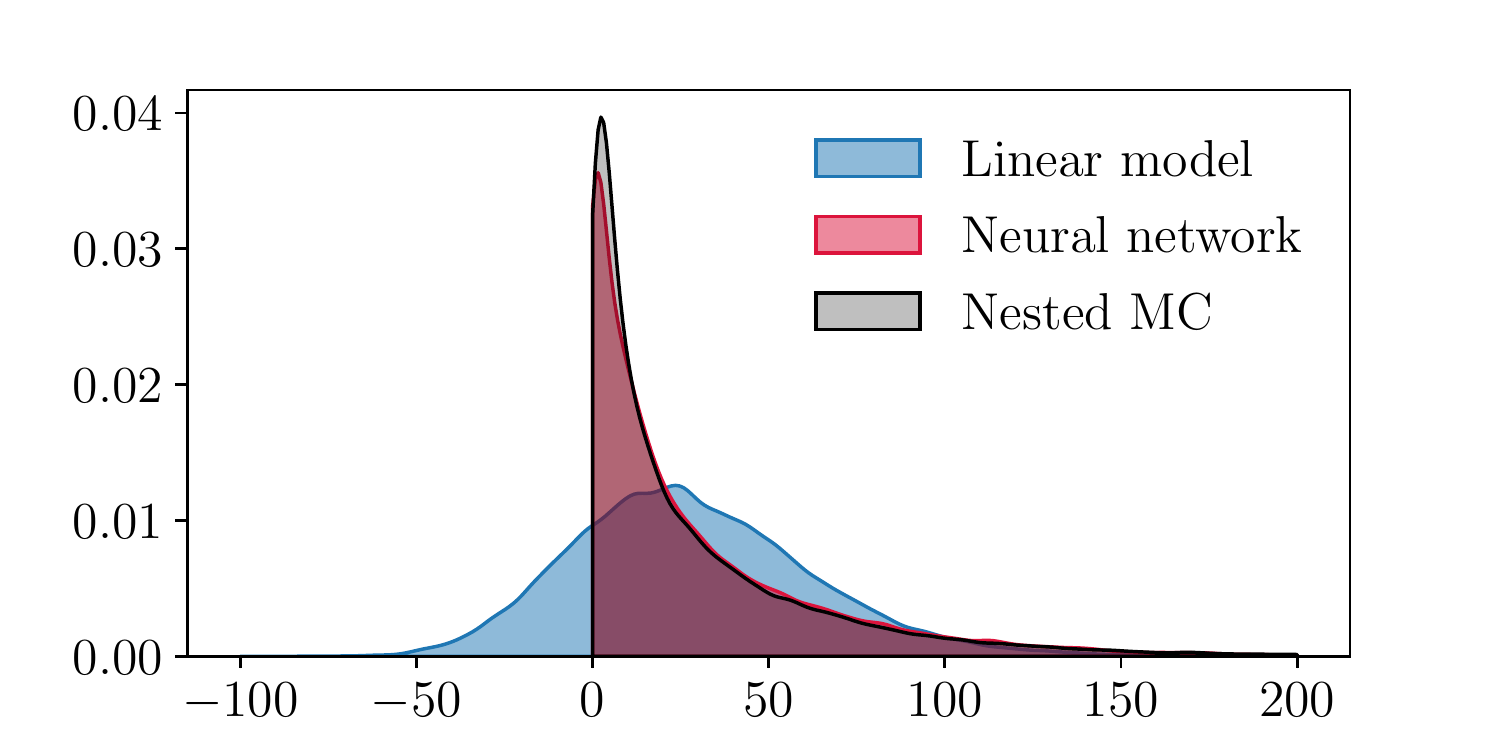}
 \caption{Density plot of 
 the
 CVA 
 of a vanilla call, at mid-life of the option.~~~~~~~~~~~~~~~~~~~~~~~}
 \label{fig:NNvsLinearCall}
\end{minipage}\hfill%
\begin{minipage}{.47\textwidth}
 \center
 \includegraphics[width=\textwidth]{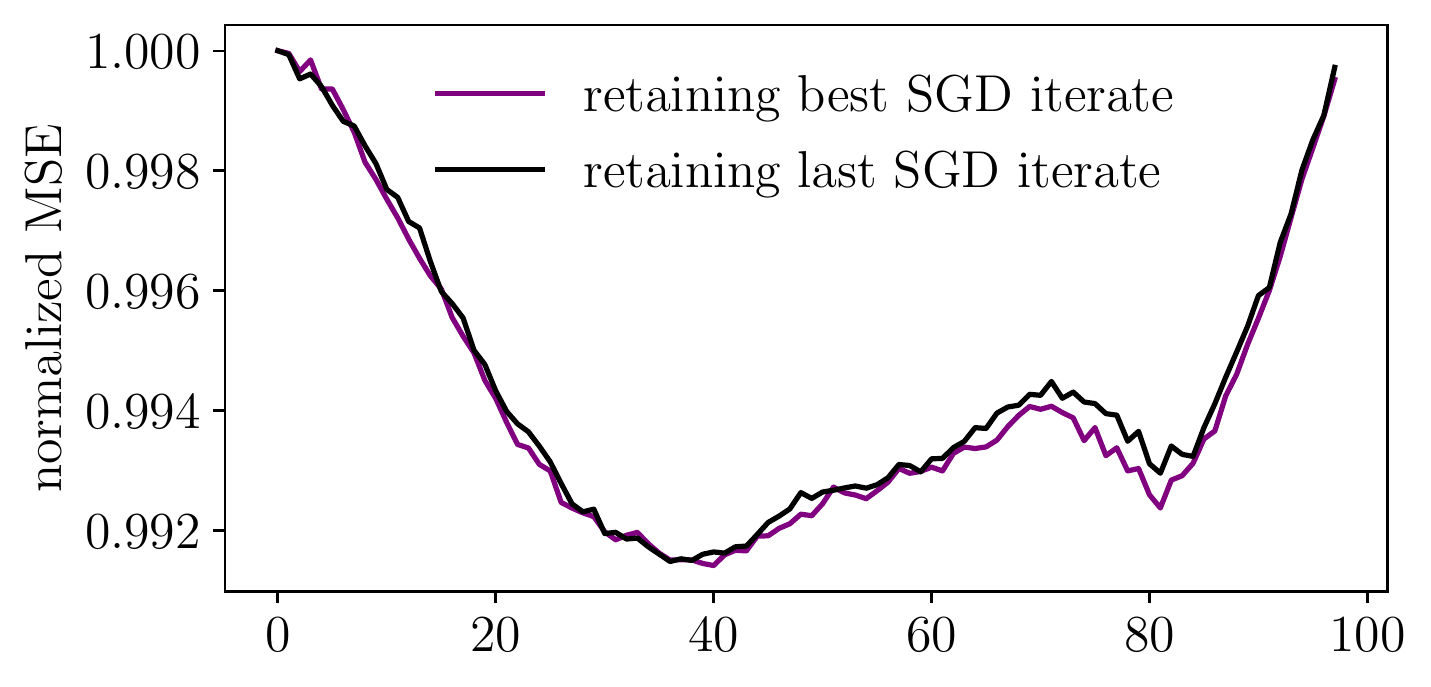}
 \caption{Out-of-sample MSEs against labels $\xi_{i,n}$ at different time-steps divided by the variance of the labels in the context of the CVA case study of Section \ref{sec:cva}.}
 \label{fig:bestvslast}
\end{minipage}
\end{figure}
In the Markov chain $X$ case, we face the additional peculiarity of a hybrid regression setting, with discrete and continuous natures of the $X$ and $Y$ model components.

Neural networks \citep{goodfellow2016deep} propose an alternative way to parameterize and learn the feature map. Let $\mathcal{NN}_{{p+q}, \theh ,u,\varsigma}$ denote the set of functions of the form
\[\mathbb{R}^{{p+q}}\owns z\mapsto \zeta(z;W^{[\theh +1]},\dots,W^{[1]},b^{[\theh +1]},\dots,b^{[1]})\eqqcolon \zeta^{[\theh +1]}(z;W,b),\] where $$W^{[\theh +1]}\in\mathbb{R}^{1\times u},\dots,W^{[\ell]}\in\mathbb{R}^{u\times u},\dots,W^{[1]}\in\mathbb{R}^{u\times (p+q)}$$ are weight matrices, $$b^{[\theh +1]}\in\mathbb{R},\dots,b^{[\ell]}\in\mathbb{R}^u,\dots,b^{[1]}\in\mathbb{R}^u$$ are bias offsets,
$W$ and $b$ are the respective concatenations of the $W^{[\ell]}$ and of the $b^{[\ell]}$,
$\varsigma$ is a scalar nonlinearity applied element-wise and, for every $z\in\mathbb{R}^{p+q}$, 
\[
\begin{aligned}
&\zeta^{[0]}(z;W,b) = z\\
 &\zeta^{[\ell ]}(z;W,b)= \varsigma(W^{[\ell ]}\zeta^{[\ell-1]}(z;W,b)+b^{[\ell]}),\quad \ell=1, \dots, h \\
& \zeta^{[\theh +1]}(z;W,b) = W^{[\theh +1]}\zeta^{[\theh ]}(z;W,b)+b^{[\theh +1]}.
\end{aligned}
\]
\sloppy {The function $z\mapsto \zeta^{[\theh  ]}(z;W,b)$ can be seen as} a nonlinear feature mapping from $\mathbb{R}^{p+q}$ to $\mathbb{R}^u$, parameterized by $W^{[\theh ]},\dots,W^{[1]}, b^{[\theh ]},\dots,b^{[1]}$ (for a given activation function $\varsigma$): see Figure \ref{fig:NN}.
\begin{figure}[htbp]
\centering
\includegraphics[scale=0.7]{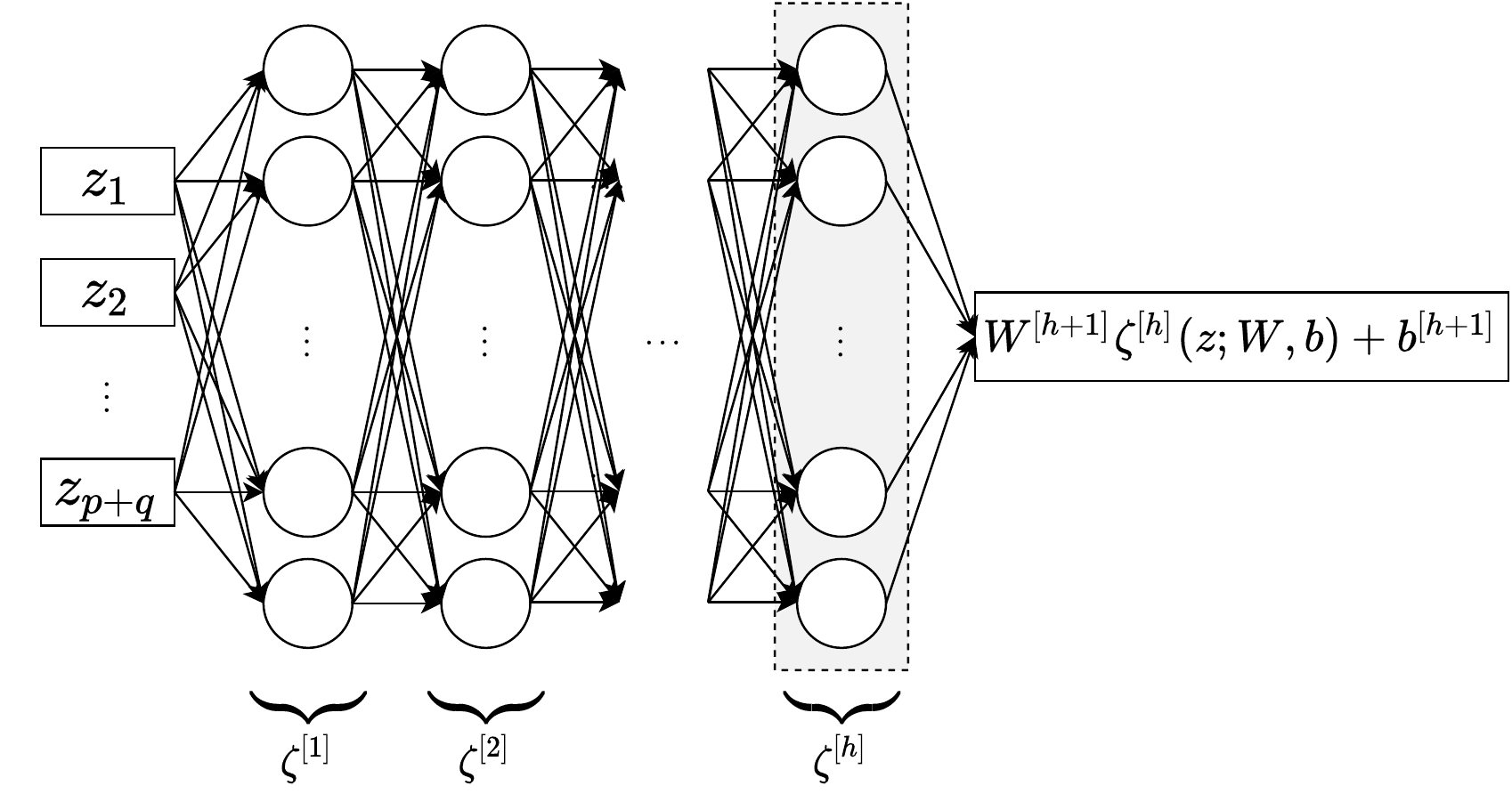}
 \caption{The last hidden layer of our network, $\zeta^{[h]}$, can be seen as a parameterization of an endogenous, implicit linear regression basis.}
 \label{fig:NN}
\end{figure}

On top of the set 
$\mathcal{NN}_{p+q, h, u,\varsigma}$ of real-valued neural networks taking inputs from $\mathbb{R}^{p+q}$, with $\theh $ hidden layers, $u$ units per hidden layer (and one output neuron), and $\varsigma$ as the activation function, 
we also define \beql{e:nnp}\mathcal{NN}_{p+q, \theh ,u,\varsigma}^+\coloneqq\{\mathbb{R}^{p+q}\owns z\mapsto(f(z))^++\mu, f\in\mathcal{NN}_{p+q, \theh , u,\varsigma}, \mu\in\mathbb{R}\}.\eeql 
This specification ensures positivity of the output when the additive constant $\mu$ is nonnegative and is useful for learning positive (e.g.~ XVA) functions. The additive constant $\mu$ is introduced in order to improve the fit of the first moment of the target function.

\subsection{Local Training Algorithm}

Learning the conditional expectation \eqref{cond_exp_discrete} in a positive neural net search space consists in applying the same empirical risk minimization
\eqref{e:forw} approximation as in linear regression, using this time $\mathcal{NN}_{p+q, \theh ,u,\varsigma}^+$ as the search space, i.e.~solving for 
 \beql{e:regrkappo}\hat{\varphi}_i\in\argmin_{\varphi_i\in\mathcal{NN}_{p+q, h,u,\varsigma}^+} \sum_{\iota\in\mathcal{I}} (\xi_{i,n}^{\iota} -\varphi_i(\theX_{\tpi }^{\iota}, \theI_{\tpi }^{\iota})
)^2 .\eeql
This is achieved by a mini-batch stochastic gradient descent.

For learning a positive output (e.g.~an XVA), the addition of a ReLU activation
$\cdot^{+}$
at the output layer in \eqref{e:nnp}
can jeopardize the learning as the gradient may vanish at a certain SGD iteration and the parameters are then frozen irrespective of the number of subsequent iterations. Thus, for more stability of the learning procedure, we first perform the first half of SGD steps on the network without the ReLU at the output layer. Then, still without the ReLU, we fine-tune the weights of the output layer by optimizing with respect to those weights only (freezing the weights of the hidden layers), which can be done in closed form in the case of quadratic risk minimization.
\brem
This fine-tuning step is not achievable in closed-form in the case of, for example, quantile regression\footnote{cf.~Remark \ref{rem:elici}.}. However, even in this case, the optimization problem is still convex and as such easier to solve numerically. 
\erem 
\noindent
Finally, we restore the ReLU at the output layer to proceed with the second half of the SGD iterations.

We also chose to retain the best set of parameters among those explored during the SGD iterations. Figure~\ref{fig:bestvslast} shows the corresponding improvement in generalization when applied in the context of the CVA case study of Section \ref{sec:cva}.

The ensuing learning scheme is detailed in Algorithm~\ref{alg:baseline_algo}. Note that we presented vanilla SGD iterations only for the sake of simplicity. In practice, accelerated SGD methods like Adam (\cite{kingma2014adam}) are used instead.
 
\begin{algorithm}[h]
\small
\SetAlgoLined
\SetKwInOut{AlgName}{name}
\SetKwInOut{Input}{input}\SetKwInOut{Output}{output}
\AlgName{BaseAlg}
\Input{$\{(X_i^{\iota}, Y_i^{\iota}, \xi_{i,n}^{\iota}),  \iota\in\mathcal{I}\}$, a partition $\mathcal{B}$ of $\mathcal{I}$, a number of epochs $E\in\mathbb{N}^{\star}$, a learning rate $\eta>0$, initial values for the network parameters $W$, $b$ and $\mu$}
\Output{Trained parameters $W_{\text{best}}$, $b_{\text{best}}$ and $\mu_{\text{best}}$} 
define $\displaystyle\mathcal{L}(W, b, \mu, \text{batch}, \text{pos})\coloneqq \left\{\begin{aligned}&\frac{1}{|\text{batch}|}\sum_{\iota\in\text{batch}} (\zeta^{[h+1]}(X^{\iota}_i, Y^{\iota}_i; W, b)+\mu-\xi_{i,n}^{\iota})^2 \quad &\text{if }\text{pos}=0\\ &\frac{1}{|\text{batch}|}\sum_{\iota\in\text{batch}} ((\zeta^{[h+1]}(X^{\iota}_i, Y^{\iota}_i; W, b))^++\mu-\xi_{i,n}^{\iota})^2 \quad &\text{if }\text{pos}=1\end{aligned}\right.$\\
 $\mathcal{L}_{\text{best}}\leftarrow \infty$, $\text{pos}\leftarrow 0$\\
 \For(\tcp*[f]{loop over epochs}){$\text{epoch} = 1, \dots, E$}{
  \For(\tcp*[f]{loop over batches}){$\text{batch} \in \mathcal{B}$}{
   \For{$\ell=1, \dots, h+1$}{
    $W^{[\ell]}\leftarrow W^{[\ell]} - \eta \nabla_{W^{[\ell]}} \mathcal{L}(W, b, \mu, \text{batch}, \text{pos})$\\
    $b^{[\ell]}\leftarrow b^{[\ell]} - \eta \nabla_{b^{[\ell]}} \mathcal{L}(W, b, \mu, \text{batch}, \text{pos})$\\
   }
   $\mu\leftarrow \mu - \eta \partial_\mu \mathcal{L}(W, b, \mu, \text{batch}, \text{pos})$\\
  }
  \If(\tcp*[f]{tune weights of last layer}){$\text{epoch}=\left\lfloor \frac{E}{2}\right\rfloor$}{
   $(W^{[h+1]}, b^{[h+1]}) \leftarrow \hspace{-3mm} \argmin\limits_{\widetilde{W}^{[h+1]}, \widetilde{b}^{[h+1]}} \hspace{-3mm}\mathcal{L}(\{W^{[0]}, \dots, W^{[h]}, \widetilde{W}^{[h+1]}\}, \{b^{[0]}, \dots, b^{[h]}, \widetilde{b}^{[h+1]}\}, \mu, obs, 0)$\\
   $\text{pos} \leftarrow 1$\\
  }
  \If(\tcp*[f]{keep track of best parameters}){$\mathcal{L}(W, b, \mu, \mathcal{I}, 1) < \mathcal{L}_{\text{best}} $}{
   $\mathcal{L}_{\text{best}}\leftarrow \mathcal{L}(W, b, \mu, obs, 1)$\\
   $W_{\text{best}}\leftarrow W$\\
  	$b_{\text{best}}\leftarrow b$\\
  	$\mu_{\text{best}}\leftarrow \mu$\\
  }
 }
\caption{Baseline learning scheme for training at a given time-step $i$}
 \label{alg:baseline_algo}
\end{algorithm}
\subsection{Backward Learning}\label{e:bd}

In the setup of the path-wise pricing problem \eqref{cond_exp_discrete},  at each pricing time $\tpi $, a separate learning problem is solved by Algorithm \ref{alg:baseline_algo}. Since the algorithm returns for each problem a local minimum, it is possible to end up with an approximation of the 
{pricing function $\mathbb{E}[\xi_{i,n} \left|\theX_{\tpi }=x, \theI_{\tpi }=y\right.]$  (cf.~\eqref{cond_exp_discrete}) } with noisy paths (i.e.~with respect to time $i$) if the local minima are not close to each other, even for fixed $x$ and $y$. Yet, for two consecutive time-steps $\tpi $ and $i+1$, the learning problems are similar. One possible refinement is, after having learned $\theE  _{i+1}$, to initialize the parameters of the
network at time $\tpi $ with the parameters of the network trained at time $i+1$. This not only smoothes the results across regression times, but also accelerates convergence.

We obtain an algorithm which starts the learnings at the last time step ${\thenb }$ and, proceeding backward in time until time step $1$, reuses each time the previously trained weights and biases as an initialization for the next learning. The ensuing backward learning scheme is detailed in Algorithm~\ref{alg:backward_algo}. This process of reusing knowledge from a different but related learning task can be seen as a form of transfer learning~\citep{pan2009survey,bozinovski2020reminder}.
{\def\imath{\text{batch}}
\begin{algorithm}[!htbp]
\small

\SetAlgoLined
\SetKwInOut{Input}{input}\SetKwInOut{Output}{output}
\Input{$\{(X_i^{\iota}, Y_i^{\iota}, \xi_{i,n}^{\iota}), \iota\in\mathcal{I}, 1 \leq i \leq n\}$, a partition $\mathcal{B}$ of $\mathcal{I}$, a number of epochs $E\in\mathbb{N}^{\star}$, 
 a learning rate $\eta>0$} 
\Output{$\hat{\varphi}_1, \dots, \hat{\varphi}_n$} 
 initialize parameters
 $W_{n+1}$, $b_{n+1}$ and $\mu_{n+1}$
 of the network at terminal time-step $n$\\
\For{$i=\thenb  \dots 1$}{
$W_i, b_i, \mu_i \leftarrow \text{BaseAlg}(\{(X_i^{\iota}, Y_i^{\iota}, \xi_{i,n}^{\iota}),  \iota\in\mathcal{I}\}, \mathcal{B}, E, \eta, W_{i+1}, b_{i+1}, \mu_{i+1})$

  $\hat{\varphi}_{\tpi}\leftarrow\{ x\mapsto \zeta^{[h+1]}(x, y; W_i, b_i)+\mu_i\}$\\ }
 \caption{Backward learning scheme}
 \label{alg:backward_algo}
\end{algorithm}
}

\brem\label{rem:bsde}
A variation on the above would be forward learning.
We favor the backward learning scheme because it is the only one that is amenable to more general backward stochastic differential equations, such as the equations for the FVA and the KVA in \citeN{Crepey21}.
In addition, in these XVA applications, the
labels/features
corresponding to times $i$ closer to the final maturity $n$ of the portfolio have a lower{/higher}  variance.  {Hence the training task corresponding to a lower time step $\tpi$ is harder}.  Proceeding backward in time, we thus solve successively harder and harder problems, but problems which benefit from the knowledge acquired solving the previous ones: a virtuous divide-and-conquer strategy.
\erem

\subsection{A Posteriori Twin Monte Carlo Validation Procedure\label{ss:twin}} 

As part of the validation of our approach, we provide an a posteriori Monte Carlo $L^2$ error estimation procedure, which can be used for assessing the quality of any estimator of a conditional expectation, without any prior knowledge on the latter (and without heavy nested Monte Carlo). At a given time step $i$, let $\xi_{i,n}^{(1)}$ and $\xi_{i,n}^{(2)}$ denote two independent copies of $\xi_{i,n}$ conditional on $(X_i, Y_i)$\footnote{The conditional independence means that for any Borel bounded functions $\phi_1$ and $\phi_2$, we have $\mathbb{E}[\phi_1(\xi_{i,n}^{(1)})\phi_2(\xi_{i,n}^{(2)})|X_i, Y_i] = \mathbb{E}[\phi_1(\xi_{i,n}^{(1)})|X_i, Y_i]\mathbb{E}[\phi_2(\xi_{i,n}^{(2)})|X_i, Y_i]$.}. For any Borel function $\varphi:\mathbb{R}^p\times\mathbb{R}^q\rightarrow\mathbb{R}$ such that $\varphi(X_i, Y_i)$ is square integrable (e.g.~a neural net estimate of $\mathbb{E}[\xi_{i,n}|X_i, Y_i]$), 
we have:
\beql{e:val1}\mathbb{E}[(\varphi(X_i, Y_i)-\mathbb{E}[\xi_{i,n}|X_i, Y_i])^2] = \mathbb{E}[\varphi(X_i, Y_i)^2-(\xi_{i,n}^{(1)}+\xi_{i,n}^{(2)})\varphi(X_i, Y_i)+\xi_{i,n}^{(1)}\xi_{i,n}^{(2)}].\eeql
The equality stems from the fact that, by   conditional independence, 
\beql{e:val2}
(\mathbb{E}[\xi_{i,n}|X_i, Y_i])^2 = \mathbb{E}[\xi_{i,n}^{(1)}|X_i, Y_i]\mathbb{E}[\xi_{i,n}^{(2)}|X_i, Y_i] = \mathbb{E}[\xi_{i,n}^{(1)}\xi_{i,n}^{(2)}|X_i, Y_i],
\eeql
followed by an application of the tower rule.
Thus, one can approximate the $L^2$ error of any estimator for the conditional expectation, without any knowledge on the latter, using only two inner paths. This can be used as a very fast validation procedure and as a safeguard in a production environment before using the learned values. A slower but more complete nested Monte Carlo approach is then only needed periodically, e.g.~after significant changes in the risk factor models, or to perform more elaborate checks (e.g.~tail behavior).


\section{\OverSimulation and its Analysis\label{s:hierarch}}

Learning tasks involving defaults such as \eqref{e:regrkappo}
may be challenging, even with optimized training schemes.  As should always be first scrutiny with machine learning, a difficulty can come from the data, \ie  from the simulation part in our case.
Specifically,
a large variance of the estimated population loss function can jeopardize the learning approach, which we address in what follows by a suitable \oversimulation approach.
 
\subsection{Identification of the Variance Contributions Using Automatic Relevance Determination\label{ss:ard}}
{\def\Xi{\nu}

In this part 
we show how to hierarchize the variance 
impact of explanatory variables using automatic relevance determination (ARD). As detailed in \citeN[Sections 5.1, 5.4.3, 6.6, and  8.3.7]{rasmussen2006gaussian}, ARD is a Bayesian procedure for feature selection and consists in estimating the relevance of the features by maximizing a marginal likelihood. 
In our case, we apply a Gaussian process regression based ARD to quantify empirically the impact of the variances of $\theX$ and $\theI$ on that of $\theY$.

Toward this aim, we treat the vector of the parameters of the data generating process (DGP, e.g.~a financial model), denoted by $\Xi$, 
as a latent variable endowed with 
some instrumental distribution.
\bex In  the CVA setup of 
Sections \ref{sec:mktcreditmodel}-\ref{sec:mktcreditmodelcontinuous},  $\nu$ is the collection of the drift and diffusion parameters of the SDEs \eqref{e:sde1}-\eqref{e:sde2} for all $e$ and $c$.
\eex
Given $\Xi$, we sample the following time-averages of the variances of $X_1, \dots, X_n$, $Y_1, \dots, Y_n$ and $\xi_{1,n}, \dots, \xi_{n,n}$:
\[
\begin{aligned}
{V(X|\Xi) } = & \frac{1}{n+1} \sum_{i=0}^n \widehat{\mathbb{E}}^\Xi [(X_i-\widehat{\mathbb{E}}^\Xi [X_i ])^2 ]\\
{V(Y|\Xi) } = & \frac{1}{n+1} \sum_{i=0}^n \widehat{\mathbb{E}}^\Xi [(Y_i-\widehat{\mathbb{E}}^\Xi [Y_i ])^2 ]\\
{V(\xi|\Xi)} = & \frac{1}{n+1} \sum_{i=0}^n \widehat{\mathbb{E}}^\Xi [(\xi_{i,n}-\widehat{\mathbb{E}}^\Xi [\xi_{i,n}])^2 ],
\end{aligned}
\]
meant componentwise in the vector cases of $X$ and $Y$, where $\widehat{\mathbb{E}}^\Xi$ is an empirical average over paths sampled for a  given realization $\Xi$ of the DGP parameters. 
Then, based on a finite sample of $\Xi$ and on the corresponding realizations of the triple $(V(X|\Xi),V(Y|\Xi) ,V(\xi|\Xi))$, we perform a Gaussian process regression \citep{rasmussen2006gaussian} of  
$V(\xi |\Xi)$ against  $V(X|\Xi)$ and $V(Y|\Xi)$.
In this procedure we use an anisotropic kernel $ 
\exp(-\sum_{j=1}^p  \frac{(v_j-v'_j)^2}{2 \lambda_{x,j} ^2}-\sum_{j=1}^q \frac{(w_j-w'_j)^2}{2 \lambda_{y,j}^2})$, where the hyperparameters $\lambda_{x,1}, \dots, \lambda_{x,p}$ and $\lambda_{y,1}, \dots, \lambda_{y,q}$ 
are characteristic length-scales for the corresponding components of $V(X|\Xi)$ and $V(Y|\Xi)$. Maximizing the marginal likelihood on the dataset allows to recover those length-scales and these can then be interpreted as relevance estimates for the different input variables. The higher the inverse length-scale gets, the more the corresponding variable influences the output (payoff variance, in our case).

This procedure is then
itself randomized, i.e.
run multiple times on the restricted datasets corresponding to different sub-samplings of $\Xi$.
This provides a distribution of the fitted hyper-parameters
$\lambda$ in the above, while being also less prone to over-fitting and local minima issues. A similar analysis was used in \citeN{bergstra2012random} to study the relevance of different neural network hyper-parameters with respect to the validation loss. 
\bex Figure~\ref{fig:ard-boxplot} reveals the dominance of the impact of the variance of $X$ on that of $\xi$ in the context of the CVA case study of Section \ref{sec:cva} (here for a single-client portfolio of the bank), where $Y$ is the vector of mark-to-market risk factor processes and $X$ 
is the default indicator process 
of the client.
\eex
\begin{figure}[h]
\begin{minipage}[b]{.36\linewidth}
\caption{\label{fig:ard-boxplot} CVA case with a single client in Section \ref{ss:leCVA}: Box-plot of the inverse length-scales obtained by randomized Gaussian process regressions of  the conditional variances of the cash flows $\xi$ against the conditional variances of the risk factors $X$ (default indicator process of the client) and $Y$, where conditional here is in reference 
to the parameters of the financial model  \eqref{e:sde1}-\eqref{e:sde2} for all $e$ and $c$,
treated as a random vector with a postulated distribution.\vspace{2cm}}
\end{minipage}
\hspace{7mm}\begin{minipage}[b]{.58\linewidth}
 \includegraphics[scale=0.63]{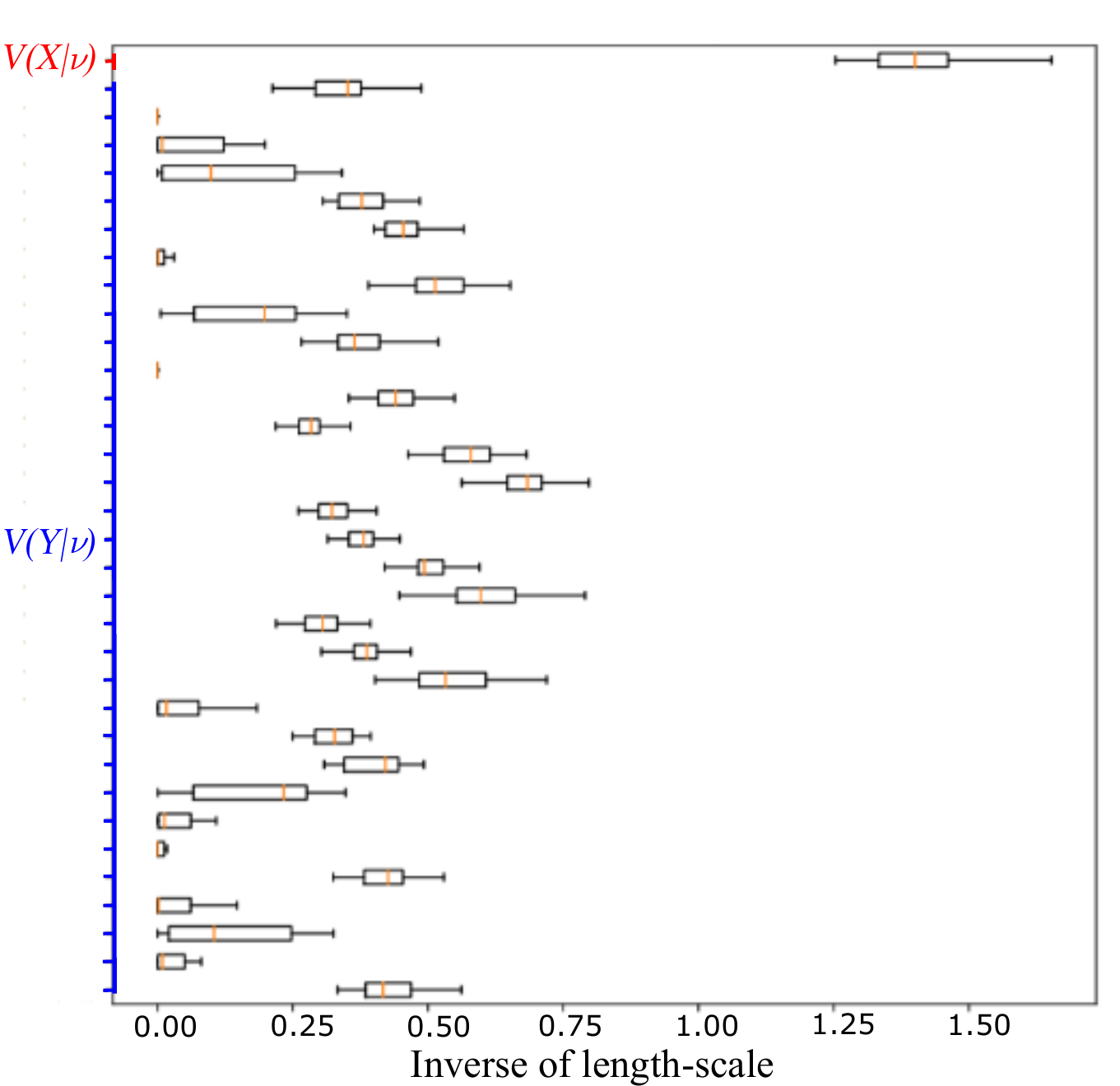} \vspace{3mm}
\end{minipage}
\end{figure}

}
\subsection{Learning on Hierarchically Simulated Paths}

If $X$ contributes more to the variance of $\xi$ than $Y$,
then, in order to deal with the resulting variance issue regarding the associated simulation/learning scheme, an idea is to simulate more realizations of $X$ than $Y$, even if this means giving up the independence of the simulation setup. More precisely, we simulate $\my$ i.i.d paths $Y^1, \dots, Y^{\my}$ of $\theI$ and, for every
$k\in \{1, \dots, \my\}$ and $i \in \{1, \dots n\}$, we simulate $\mx$ i.i.d realizations 
 $X^{k, 1}_i, \dots, X^{k, \mx}_i$ of $X_{i}$ conditional on $\theI^k$.
For every $i$, this yields a sample $(\theX_i^{k, l}, Y_i^k, \xi_{i,n}^{k, l})\sp k\in \{1, \dots, \my\}, l\in\{1, \dots, \mx\}$ of $(X_i, Y_i, \xi_{i,n})$ of size $\my\mx$, where, within each block $k$, independence between the $\theX^{\lefta k, l   \righta}$  only holds 
conditionally on $\theI^k$.  


Algorithm \ref{alg:backward_algo} is then run on the resulting hierarchically simulated dataset by taking $\mathcal{I}=\{1, \dots, \my\} \times \{1, \dots, \mx\}$, with  $Y^{k,l}:=Y^k$ for all $l$.
For implementation efficiency reasons pertaining to memory contiguity, 
the set of indices of the $\imath$-th batch, with $1\leq\imath\leq|\mathcal{B}|$, is chosen to be $\{(k, l)\in \{1, \dots, \my\} \times \{1, \dots, \mx\}: (\imath-1)\frac{|\mathcal{I}|}{\left|\mathcal{B}\right|} \leq (l-1)\mx+(k-1) + 1 < \imath\frac{|\mathcal{I}|}{\left|\mathcal{B}\right|}\}$.


\Oversimulation in the above sense can be thought of as a form of
data augmentation procedure \citep{shorten2019survey}, but in a simulation setup where one knows how to generate the data perfectly.
The main question is then to which extent one should augment the data, i.e.~the choice of the \oversimulation parameters $\my$ and $\mx$, which is the focus of the sequel of this section. 
\brem \label{rem:IS}  \Oversimulation is 
different in nature from importance sampling that favors particular events, e.g., in a credit risk setup, default versus survival (see e.g.~\citeN{carmona2010particle}).
In an XVA setup, some metrics, like the~CVA, need default events for being properly estimated, whereas others, like the FVA, require survival events. Hence what one needs is richness regarding both default and survival events, which is what \oversimulation provides.

It is also unrelated to nested Monte Carlo as per \citeN{Gordy} or (in an XVA setup) \citeN{abbas2018xva}, which optimizes the complexity of the  simulation with respect to the number of trajectories on each Monte Carlo layer. In these contributions, the inner conditional trajectories are not assumed to be much less complex to simulate than the outer one.
Thus, in contrast to the asymptotic results obtained for all layers in \citep*{Gordy, abbas2018xva}, Heuristic \ref{heu:2} separates the asymptotic strategy for $N$ from the non asymptotic strategy proposed for $M$.
~\erem

\subsection{Choosing the \OverSimulation Factor\label{ss:choice1}}

Assume that simulating $Y_i$ costs $\kappa$ times more than simulating $X_i$  given a path $\{Y_j\}_{j \leq i}$
in terms of computation time. The \oversimulation factor $\mx$ can be chosen so as to minimize the variance ($\tmop{\Var}$) of the loss $\frac{1}{\my \mx} 
\sum_{k = 1}^{\my} \sum_{l = 1}^{\mx} g_i(\theta, X_i^{k, l}, \dots, X_n^{k, l}, Y_i^k, \dots, Y_n^k)$ with respect to $\mx$, under a budget constraint $ \my (\mx + \kappa)=\Psi $, where
$g_i$ is the point-wise loss of our learning task at time-step $i$, e.g.~$g_i (\theta, X_i, \dots, X_n, Y_i, \dots, Y_n) = (\xi_{i,n} -  \varphi^{\theta}_i  (X_i, Y_i))^2$, where $  \varphi^{\theta}_i $ is a neural net (element of $\mathcal{NN}_{p+q, \theh ,u,\varsigma}^+$) with
 parameters collectively denoted by $\theta$ (cf.~\eqref{e:regrkappo}).

For ease of notation in this and the next part, 
we write $g_i(\theta, X^{k, l}, Y^k)$ and $g_i(\theta, X, Y)$ instead of $g_i(\theta, X_i^{k, l}, \dots, X_n^{k, l}, Y_i^k, \dots, Y_n^k)$ and $g_i(\theta, X_i, \dots, X_n, Y_i, \dots, Y_n)$ (it is then implied that $X$ and $Y$ play formally the role of vectors containing their path from time-step $i$ up to $n$).

\bp\label{p:heurisN}
The \oversimulation factor that minimizes the variance of the loss $\frac{1}{\my \mx} 
\sum_{k = 1}^{\my} \sum_{l = 1}^{\mx} g_i(\theta, X^{k, l}, Y^k)$ with respect to $\mx$, subject to the budget constraint $ \my (\mx + \kappa)=\Psi ,$ 
 is  \begin{equation}\label{e:optom}\mx^\theta_i =\sqrt{\frac{Q^{\theta}_i \kappa}{R^{\theta}_i}},\end{equation}
where
\bel&
R^{\theta}_i=\mathbb{C}{\rm ov}\big(g_i(\theta, X^{1, 1}, Y^1),g_i(\theta, X^{1, 2}, Y^1)\big)=\mathbb{V}{\rm ar}\big(   \mathbb{E} (g_i(\theta, X^{1, 1}, Y^1) | Y^1) \big)
\\&
Q^{\theta}_i  = \mathbb{E}\big( \mathbb{V}{\rm ar} (g_i(\theta, X^{1, 1}, Y^1) | Y^1)\big)= \mathbb{V}{\rm ar} (g_i(\theta, X^{1, 1}, Y^1)  ) - R^{\theta}_i.
\eel 
\ep
\proof 
After rearranging terms, one can show that
\[
 \tmop{\Var} ( \frac{1}{\my \mx}  \sum_{k = 1}^{\my} \sum_{l =
   1}^{\mx} g_i(\theta, X^{k, l}, Y^k) ) = \frac{R^{\theta}_i}{\Psi}  ( \frac{1}{\mx} ( \mx - \sqrt{\frac{Q^{\theta}_i \kappa}{R^{\theta}_i}})^2 + ( \sqrt{\frac{Q^{\theta}_i}{R^{\theta}_i}} + \sqrt{\kappa} )^2 ),  
\]
where
\bel&
Q^{\theta}_i =
 \mathbb{E} [(g_i(\theta, X^{1, 1}, Y^1))^2] -\mathbb{E} [g_i(\theta, X^{1, 1}, Y^1) g_i(\theta, X^{1, 2}, Y^1)]
\\&
R^{\theta}_i
  =\mathbb{E} [g_i(\theta, X^{1, 1}, Y^1) g(\theta, X^{1, 2}, Y^1)] - (\mathbb{E} [g_i(\theta, X^{1, 1}, Y^1)])^2.~\finproof
\eel

The quotient \bel&\frac{Q^{\theta}_i}{R^{\theta}_i}    =\frac{\mathbb{E}\big( \mathbb{V}{\rm ar} (g_i(\theta, X^{1, 1}, Y^1) | Y^1)\big)}{\mathbb{V}{\rm ar}\big(\mathbb{E} (g_i(\theta, X^{1, 1}, Y^1) | Y^1)\big)} \eel in \eqref{e:optom} measures the 
relative
contributions of $Y$ and $X$
to the variance of the loss estimator (note that $Q^{\theta}_i+R^{\theta}_i=\mathbb{V}{\rm ar}  (g_i(\theta, X^{1, 1}, Y^1)) $, by the total variance formula). 
To estimate the values of $Q^{\theta}_i$ and of $R^{\theta}_i$ therein, one only needs to simulate $(X^{1, 1}, X^{1, 2}, Y^1)$, i.e., with respect to the bare simulation of $(X, Y)$, one extra simulation of $X$ conditional on each realization of $Y$.

As a fixed value of $\mx$ has to be chosen throughout all the simulation and training task, for the above result to be of practical use, $\mx^\theta_i$ has to be reasonably stable with respect to both pricing time steps $i$ and SGD iterations (the transfer learning scheme of Section~\ref{e:bd} is advantageous in this respect in that it stabilizes the learning). If so, it leads to the following:

\begin{heuristic} \label{heur:MN}
Choose for $N$ the average 
 of the values $N_i^{\theta}$ obtained during the SGD iterations and the time steps. Make for $M$ the corresponding choice deduced from the budget constraint, i.e.~$M
 =\frac{\Psi}{N+P}$. 
\end{heuristic}
\noindent
Note that $N$ depends only on $P$, and $M$ on $P$ and $\Psi$.  
If $P$ is not analytically known, it can be deduced from simulation times of experiments corresponding to the same $M$ but different $N$.
Namely,
let $\Psi$ and $\Psi'$ the budgets corresponding to configurations $(M, N)$ and $(M, N')$. We have \beql{l:ratio}
\frac{\Psi}{\Psi'}=\frac{P+N}{P+N'}.\eeql 
One can deduce $P$ by identifying the ratio in \eqref{l:ratio} to that of the execution times of $(M, N)$ and $(M, N')$. For doing so, it is preferable to choose $M$ large enough to avoid time measurement noise that may be due to caching or parallelization of the simulations.


\subsection{Statistical Convergence Analysis \label{ss:statconvanalysis}}

In this part we completely omit the index $i$ from the notation.
For every possible parameterization $\theta \in \Theta \subset \mathbb{R}^d$ of our neural network (with $d$ parameters), define
\begin{eqnarray}
  G (\theta) & \coloneqq & \mathbb{E} [g (\theta, X, Y)] \nonumber\\
  \hat{G}_{\my , \mx  } (\theta) & \coloneqq & \frac{1}{\my\mx}  \sum_{k = 1}^\my \sum_{l =
  1}^\mx g (\theta, X^{k, l}, Y^k) \nonumber
\end{eqnarray}
and, for all $\epsilon > 0$ and non-empty subsets $E$ of $\Theta$:
\begin{eqnarray}
  S^{\epsilon} (E) & \coloneqq & \{ \theta \in E : G (\theta) \leq
  \min_{E} G  
  + \epsilon \} \nonumber\\
  \hat{S}^{\epsilon}_{\my , \mx  } (E) & \coloneqq & \{\theta \in E : \hat{G}_{\my , \mx}
  (\theta) \leq
 \min_{  E} \hat{G}_{\my , \mx  }.
   + \epsilon\} .
\end{eqnarray}
Let $\mathcal{X}$ and $\mathcal{Y}$ be the state spaces of $X$ and $Y$.
For all
$\theta, \theta' \in \Theta$ and $t \in \mathbb{R}, y \in \mathcal{Y}$,
denote:
\begin{eqnarray}
  \Gamma (\theta, \theta', X, Y) & \coloneqq & g (\theta', X, Y) - g (\theta, X,
  Y) \nonumber\\
  \mathcal{M} (\theta, \theta', t, y) & \coloneqq & \mathbb{E} [\exp (t \Gamma
  (\theta, \theta', X, Y)) |Y = y]. \label{e:Mt}
\end{eqnarray}
We are interested in the event
\beql{e:theevent} \{ \hat{S}^{\delta}_{\my , \mx} (E) \not\subset S^{\epsilon} (E) \}
   = \bigcup_{\theta \in E \setminus S^{\epsilon} (E)} \bigcap_{\theta' \in E}
   \{ \hat{G}_{\my , \mx} (\theta) \leq \hat{G}_{\my, \mx} (\theta') + \delta \} ,\eeql
where $E$ is a non-empty subset of $\Theta$ and $\epsilon, \delta > 0$.
This is the event that a close to minimum of the finite-sample problem is far from being a minimum of the population mean minimization problem.
Theorem \ref{bound_finite} provides a bound on the probability of this
event when $E$ is finite.

\begin{theorem}
  \label{bound_finite}
  Let $E$ be a finite and non-empty subset of $\Theta$ and
  let $0 < \delta < \epsilon$. Assume that $E \setminus S^{\epsilon} (E) \neq
  \emptyset$ and that there exist $b_1, b_2 > 0$ such that for every $t \in
  \mathbb{R}$ and $\theta, \theta' \in E:$  
  \begin{eqnarray}
    \mathcal{M} (\theta, \theta', t, Y) & \leq & \exp ( \mathbb{E}
    [\Gamma (\theta, \theta', X, Y) |Y] t + \frac{b_1^2 t^2}{2} ) \quad
    \text{a.s.}  \label{subgaussian1}\\
    \mathbb{E} [\exp (t\mathbb{E} [\Gamma (\theta, \theta', X, Y) | Y])] &
    \leq & \exp ( \mathbb{E} [\Gamma (\theta, \theta', X, Y)] t +
    \frac{b_2^2 t^2}{2} ) . \label{subgaussian2}
  \end{eqnarray}
  Then
  \beql{e:impactN1}  \Q ( \hat{S}^{\delta}_{\my, \mx} (E) \not\subset S^{\epsilon}
     (E) ) < | E | \exp ( \frac{- \my(\epsilon - \delta)^2}{2 (
   b_1^2/\mx + b_2^2 )} ). \eeql
\end{theorem}

\begin{proof}
  See Section \ref{proof_bound_finite}.~\finproof\\
\end{proof}

Theorem \ref{bound_infinite} yields a similar bound valid for a possibly infinite parameter space, 
under additional assumptions of compactness and convexity of this space and
Lipschitz continuity of the point-wise loss function. For brevity we write $S^{\epsilon}(\Theta) \coloneqq S^{\epsilon} $, $\hat{S}^{\delta}_{\my, \mx}(\Theta) \coloneqq \hat{S}^{\delta}_{\my, \mx}$.

\begin{theorem}
  \label{bound_infinite}Assume that $\Theta$ is a compact and convex and let $0 < \delta < \epsilon$. Let $D \coloneqq \sup_{\theta,
  \theta' \in \Theta} \| \theta - \theta' \|$ and assume that there exists a
  mapping $L : \mathcal{X} \times \mathcal{Y} \rightarrow
  \mathbb{R}^{\star}_+$ such that $\mathbb{E} [\exp (t L(X, Y))] <
  \infty$ for all $t$ in some neighbourhood of $0$ and for all $\theta,
  \theta' \in \Theta$:
  \begin{equation}
    \label{lipschitz} | g (\theta, X, Y) - g (\theta', X, Y) | \leq L(X,
    Y)  \| \theta - \theta' \| \quad \text{a.s.}
  \end{equation}
  Let $\bar{L} \coloneqq \mathbb{E} [L(X, Y)]$ and assume that there exist
  $\ell_1, \ell_2 > 0$ such that:
  \begin{eqnarray}
    | L(X, Y) -\mathbb{E} [L(X, Y) |Y] | & \leq & \ell_1 \quad
    \text{a.s.}  \label{lipschitz_bounded1}\\
    | \mathbb{E} [L(X, Y) |Y] - \bar{L} | & \leq & \ell_2 \quad \text{a.s.} 
    \label{lipschitz_bounded2}
  \end{eqnarray}
  and that there exist $b_1, b_2 > 0$ such that for every $t
  \in \mathbb{R}$ and $\theta, \theta' \in \Theta$ :
 \begin{eqnarray}\label{subgaussian21}&
    \mathcal{M} (\theta, \theta', t, Y)   \leq  \exp ( \mathbb{E}
    [\Gamma (\theta, \theta', X, Y) |Y] t + \frac{b_1^2 t^2}{2} ) \quad
    \text{a.s.}  \\&\label{subgaussian22}
    \mathbb{E} [\exp (t\mathbb{E} [\Gamma (\theta, \theta', X, Y) | Y])]  
    \leq  \exp ( \mathbb{E} [\Gamma (\theta, \theta', X, Y)] t +
    \frac{b_2^2 t^2}{2} ). 
     \end{eqnarray}
  Then 
\beql{e:impactN2} \Q ( \hat{S}^{\delta}_{\my, \mx} \not\subset S^{\epsilon}
     ) \leq \inf_{L' > \bar{L}} \{  \big( ( \frac{8 L' D}{\epsilon -
     \delta}+1 )^d +1\big)\exp ( \frac{- \my(\epsilon - \delta)^2}{8 (
     b_1^2/\mx + b_2^2 )} ) + \exp ( \frac{- \my(L' -
     \bar{L})^2}{2 ( \ell_1^2/\mx + \ell_2^2 )} ) \}. \eeql
\end{theorem}

\begin{proof}
  See Section \ref{proof_bound_infinite}.~\finproof\\
\end{proof}

\noindent
The Lipschitz assumptions of Theorem~\ref{bound_infinite} are reasonable in our case since our neural network is Lipschitz with respect to its parameters, and its composition with the loss function remains Lipschitz if we assume that the parameters are bounded. In particular, these Lipschitz assumptions are satisfied in our learnings if we assume that {\em (i)} the processes $X$ and $Y$ are bounded (natively or after numerical truncation), {\em (ii)} the payoff function $f$ (cf.~\eqref{cond_exp_discretebis}), which is embedded in the loss function $g$, is Lipschitz continuous or bounded, and {\em (iii)} Lipschitz continuous activation functions are used in the neural networks.
%


The following result can help in selecting $M$ for reaching a target confidence level $(1-\alpha)$.
\bcor
\label{ss:choice2}
Let $0 < \alpha < 1$ and assume the conditions of Theorem~\ref{bound_infinite}. Choose any weight $u\in(0,1)$ and $L'>\bar{L}$. Then, for
\beql{e:optiM}  &\my = \max\{ \frac{8(b_1^2/\mx+b_2^2)}{(\epsilon-\delta)^2}\log(\frac{1}{u \alpha}\big((\frac{8 L' D}{\epsilon - \delta}+1)^d+1\big)), \frac{8(\ell_1^2/\mx+\ell_2^2)}{(L'-\bar{L})^2}\log(\frac{1}{(1-u) \alpha}) \},\eeql
we have \beql{e:pac}\mbox{$\hat{S}^{\delta}_{\my , \mx  } \subset S^{\epsilon}$ with probability at least $(1-\alpha)$.}\eeql
\ecor

\begin{proof}
By Theorem \ref{bound_infinite}, for \eqref{e:pac} to hold, it suffices that
\bel
&\big( ( \frac{8 L' D}{\epsilon - \delta}+1 )^d+1\big) \exp ( \frac{- \my(\epsilon - \delta)^2}{8 (\frac{b_1^2}{\mx  } + b_2^2 )} )  <  u\alpha \mbox{ and }\\
&\exp ( \frac{- \my(L'-\bar{L})^2}{2 ( \frac{\ell_1^2}{\mx  } + \ell_2^2 )} )  <  (1-u)\alpha,
\eel
which is verified by choosing $\my$ as  \eqref{e:optiM}.~\finproof\\
\end{proof}

\noindent
As the formula \eqref{e:optiM} for $M$ is decreasing in $N$, no matter how large $N$ is, $M$ has to be greater than the limit of \eqref{e:optiM} as $N\rightarrow\infty$, which provides a lower bound for $M$. This is natural as we do not expect to get an efficient sampling and good generalization just by increasing the number of realizations of $X$ only.

Corollary~\ref{ss:choice2} suggests the following:
\begin{heuristic}\label{heu:2}
In order to satisfy a constraint $\Q(\hat{S}^{\delta}_{\my , \mx  } \subset S^{\epsilon}) \geq 1-\alpha$ (instead of a target budget $\Psi$ in  Heuristic \ref{heur:MN}),  choose $N$ as in Heuristic~\ref{heur:MN} (since it is independent of the budget value),   then deduce $M$ from the formula \eqref{e:optiM}.
\end{heuristic}
\noindent

\noindent
In the data augmentation mindset framed before Remark \ref{rem:IS}, one could then set the size $\my$ of the market data $Y$ as a function of the \oversimulation factor $\mx$ and of the confidence level $(1-\alpha)$. In a context where collecting the market data $Y$ is expensive,  Heuristic \ref{heu:2} would thus allow the user
to 
benefit from the augmentation factor $\mx$ through a reduction of the size
$\my$ of the dataset for $Y$. However, making Heuristic \ref{heu:2} really practical would require to estimate the parameters $b_1$, $b_2$, $\ell_1$, $\ell_2$ and $\bar{L}$ in Corollary~\ref{ss:choice2}. 

\section{CVA Case Study\label{sec:cva}}

We illustrate our approach numerically by a CVA case study.  In this context the probability measure $\Q$ represents a risk-neutral measure chosen by the market, to which a model of financial risk factors is calibrated.

\subsection{Market and Credit Model\label{sec:mktcreditmodel}}
We consider a bank trading derivative contracts in different economies $e$ with various clients $c$.
The currency corresponding to the economy labeled by $0$ is taken as the reference currency. 
Let there be given the short rate process $r^{\left\langle \thee\right\rangle}$ in each economy $\thee$, as well as the exchange rate process $\chi ^{\left\langle \thee\right\rangle}$ from the currency of each economy $\thee\neq 0 $ to the reference currency. 
Each client $c$ of the bank has a stochastic default intensity process $\gamma^{\left\langle c\right\rangle}$ and a default-time $\tau^{\left\langle c\right\rangle}$. For notational convenience we also define $\chi ^{\left\langle 0\right\rangle}=1$ and we denote by $\gamma^{\left\langle 0\right\rangle}$ the default intensity of the bank itself.
We consider an Euler-Marayama time-discretization of the model in Section \ref{sec:mktcreditmodelcontinuous}. We 
use the same notation for the continuous-time processes and
their discrete-time approximations (with time-step equal to $1\ \text{year}$ to alleviate the notation).

\brem\label{rem:finecoarse} 
In practice, the time-discretizations are stepping through a refined simulation time grid. This simulation grid is also used when integrating numerically some of the above diffusions, e.g.~the default intensities in \eqref{e:tau}, 
or for defining risk-neutral discount factors ${\beta}_{\tpi }$
associated with the reference currency by approximating $(-\ln{\beta}_{\tpi })$ using numerical integration\footnote{It is also possible to jointly simulate exactly $r^{\left\langle 0\right\rangle}$ and its integral without the need for numerical integration, see for example \citeN{glasserman2004monte}.} of $r^{\left\langle 0\right\rangle}$ on $[0, i]$.
Learning, pricing and checking for default events, instead, are only done at the coarser pricing time steps. Hence, although we step through the fine time grid in our discretized diffusions, we only need to store the values of the processes at the pricing time steps.
 \erem

We 
define $X$ as the collection of all the default indicator processes of the clients $c$
and $Y$ as the collection of all the short interest rate, FX, and default intensity processes $r$, $\chi$ and $\gamma$ (except for the instrumental $\chi ^{\left\langle 0\right\rangle}=1$), endowed with the filtration generated by the innovation in the model, i.e.~the collection of all the
Gaussian
and exponential variables involved at the increasing time steps $i\in\{1, \dots, n\}$. Note that both $Y$ (by itself) and $(X,Y)$ (jointly) are Markov processes with respect to this filtration.


\subsection{Learning the CVA}\label{ss:leCVA}

We denote by $\mathrm{MtM}^{\left\langle c\right\rangle}_i$ the mark-to-market at time $i$, from the point of view of the bank and in units of the reference currency, of all the contracts with the client $c$. By mark-to-market we mean trade additive counterparty-risk-free valuation, i.e.~the risk-neutral conditional expectation of the future contractually promised cash flows, expressed in units of the reference currency and discounted at the risk-free rate $r^{\langle 0 \rangle}$. We restrict ourselves to interest-rate derivatives for which mark-to-market valuation at $i$ is a function of $Y_i$,
by the nature of the cash-flows and the Markov property of $Y$\footnote{see Remark \ref{rem:pathdep}.}.
The CVA of the bank then corresponds to the risk-neutral conditional expectation of its future risk-free discounted client default losses. Namely,
the  
CVA of the bank  at the time step $\si $ is given by\footnote{Assuming that the netting set for each client is the whole set of transactions with this client.} 
\begin{equation}\label{eq:CVA_portfolio}
\mathrm{CVA}_{\si } =
\mathbb{E}\Big[\sum_{c} \sum_{j=i}^n \beta_i^{-1}\beta_{j+1}({\mathrm{MtM}}_{j+1}^{\left\langle c\right\rangle})^+\mathbbm{1}_{j<\tau^{\left\langle c\right\rangle}\leq j+1}
\Big| X_i, Y_i
\Big].
\end{equation}
\noindent
Hence
$\mathrm{CVA}_{\si }=\varphi^{\star}_i(X_i, Y_i) $, where 
\begin{equation}\label{eq:onemore}\varphi^{\star}_i \in \argmin_{\varphi \in\cB( \mathbb{R}^{p+q})} \mathbb{E}\big[(\sum_{c}\sum_{j=i}^n \beta_i^{-1}\beta_{j+1}({\mathrm{MtM}}_{j+1}^{\left\langle c\right\rangle})^+\mathbbm{1}_{j <\tau^{\left\langle c\right\rangle}\leq j+1}-\varphi(X_i, Y_i))^2\big]  .\end{equation}

We also mention the following default intensity-based formula for the CVA of the bank (cf.~\citet[Eq.~(60)]{CrepeyHoskinsonSaadeddine2019}):
\begin{equation}\label{eq:CVAint}
\widetilde{\mathrm{CVA}}_i = 
\mathbb{E}\Big[\sum_c \sum_{j=i}^{ {\thenb-1}} {\beta}_{i }^{-1}{\beta}_{j}({\mathrm{MtM}}_{j}^{\left\langle c\right\rangle})^+\gamma^{\left\langle c\right\rangle}_{j}
\exp\big(-\sum_{s=i}^ {j-1} \gamma^{\left\langle c\right\rangle}_s
 \big) \mathbbm{1}_{\{\si <\tau^{(c)} \}}
\Big| X_i, Y_i
\Big], 
\end{equation}
which converges to the same continuous-time limit as $\mathrm{CVA}_i$ when the time discretisation
  step\footnote{Conventionally set to 
  one
  in this paper.}
  goes to zero. Hence
$\widetilde{\mathrm{CVA}}_i= \tilde{\varphi
}
_i
(X_i, Y_i) $, where {(cf.~\eqref{e:genregr})}
\beql{e:onecliint} \tilde{\varphi}
_i
 \in \argmin_{\varphi \in\cB( {\mathbb{R}^{p+q}} )} \mathbb{E}
  \Big[\sum_c \sum_{j=i}^{{\thenb-1} } {\beta}_{\tpi }^{-1}{\beta}_{\tjp  }({\mathrm{MtM}}_i^{\left\langle c\right\rangle})^+
\gamma^{\left\langle c\right\rangle}_{j} 
\exp(-\sum_{s=i}^ {j-1} \gamma^{\left\langle c\right\rangle}_s 
 ) \mathbbm{1}_{\{\si <\tau^{(c)} \}}   
-\varphi(X_i,Y_i)\Big]^2 
.\eeql  
  
We reiterate that Algorithm~\ref{alg:backward_algo} with \oversimulation of  $(X,Y)$
  is generically applicable to all the XVA metrics.
The focus on the CVA in our case study is for benchmarking purposes. Were it for the CVA only,
  the regression learning scheme with minimal variance is obviously the one based on \eqref{eq:CVAint}. Also note that learning the CVA client by client, exploiting the linearity of the conditional expectation for this purpose, would lead to  as many regressions as there are clients of the bank, which for realistic banking portfolios would be extremely inefficient.
  

On top of the learning schemes \eqref{e:onecliint} and \eqref{eq:onemore} associated with the formulations
\eqref{eq:CVAint} and \eqref{eq:CVA_portfolio}, another computational alternative
in each case is nested Monte Carlo as detailed in \citeN{abbas2018xva}. This variety of approaches is useful for benchmarking purposes.
 Specifically, we implemented the learning procedure of Algorithm \ref{alg:backward_algo} in PyTorch with custom CUDA kernels for label generation during the backward iterations, the way detailed in Section \ref{ss:pycuda} (and in the accompanying github repository). In addition we implemented an optimized CUDA benchmark involving nested simulations, using the intensity-based formulation \eqref{eq:CVAint} for the inner CVA computations.
For the nested Monte Carlo,
in consideration of
the square-root rule recalled in \citeN[Section 3.3]{abbas2018xva}, 
we used $128$ inner paths. The nested Monte Carlo CVA is only computed at few pricing times due to the heavy calculation.

\subsection{Preliminary Learning Results Based on IID Data\label{sec:res}}

In the following experiments, we assume that the bank is trading derivatives in $10$ economies with $8$ clients. Implementing the discretized market and default model, we get a total of 10 interest rates, 9 cross-currency rates, and 8 default intensities. 
This yields $27$ diffusive market risk factors and $8$ default indicator processes. For time-stepping, we use $n=100$ pricing time steps and 
$25$ simulation sub-steps per pricing time step (see Remark~\ref{rem:finecoarse}). We consider a portfolio of $500$ interest rate swaps with random characteristics (notional, currency and counterparty), the MtM$^{\left\langle c\right\rangle}$ are thus analytic. All swaps are priced at par at inception. For all the runs of the simulations in this section, whether they be for training or testing, we use $\my=16384$ paths for the market risk factors $Y$.

The comparison between the two panels of 
Figure \ref{fig:intensity_vs_indicator} reveals a difficulty with the neural net learning approach of Algorithm \ref{alg:backward_algo} applied to the  defaults-based formulation \eqref{eq:CVA_portfolio} on the basis of i.i.d. simulated data. In this case, represented by the left panel in Figure \ref{fig:intensity_vs_indicator}, 
the network only learns a rather crude and noisy approximation of the CVA conditional to each training time: it is only
on the mean that 
 the learned CVA  
agrees with the nested Monte Carlo estimator;
on the tails it largely fails. As visible from the right panel, the CVA learned
using the intensity-based formulation, instead, yields satisfactory results on a  wide range of quantiles of the targeted distribution. 
\begin{figure}
\centering
\includegraphics[scale=0.55]{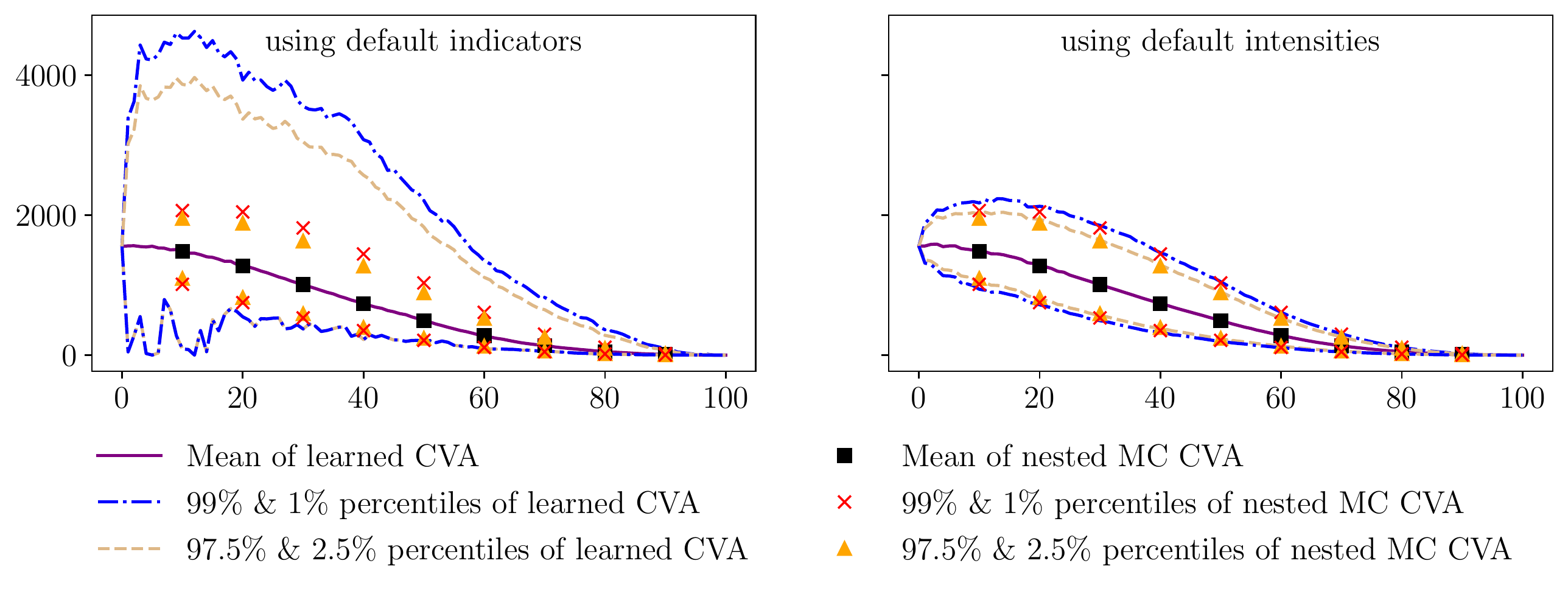}
 \caption{CVA learned using default indicators versus using default intensities ($x$ axis pricing times, $y$ axis CVA levels). Statistics computed using out-of-sample paths.}
 \label{fig:intensity_vs_indicator}
\end{figure}

\subsection{Learning Results Based on Hierarchically Simulated Data\label{sec:cva2}}
 
\def\thene{Na_{\text{rates}}} \def\thene{n_e}
\def\thenc{\thenc }\def\thenc{n_c}
In order to improve the learning \eqref{eq:onemore} of the defaults-based CVA \eqref{eq:CVA_portfolio}, we apply to it the hierarchical simulation technique of Section \ref{s:hierarch}.
Let $(r^1 ,\chi^1, \gamma^1), \dots, (r^{\my} ,\chi^{\my}, \gamma^{\my})$, be i.i.d sample paths of the triple of processes  $(r,\chi,\gamma)$.
Let $\{\epsilon^{k, l}, 1\leq k\leq \my, 1\leq l\leq \mx\}$ be i.i.d samples of $\epsilon$, the vector
defined by the right-hand side in \eqref{e:tau} where $c$ ranges over clients.
Then we can define $\my\mx$ samples of the vector of the default indicator processes of the clients at every pricing time $i$ based on~\eqref{e:tau}.
Figure~\ref{fig:defaults_simulation} illustrates the ensuing simulation scheme for the default indicator of a generic client of the bank, with sampled default times $\tau^{k,l}.$

\begin{figure}[!htbp]
\centering
\includegraphics[scale=0.6]{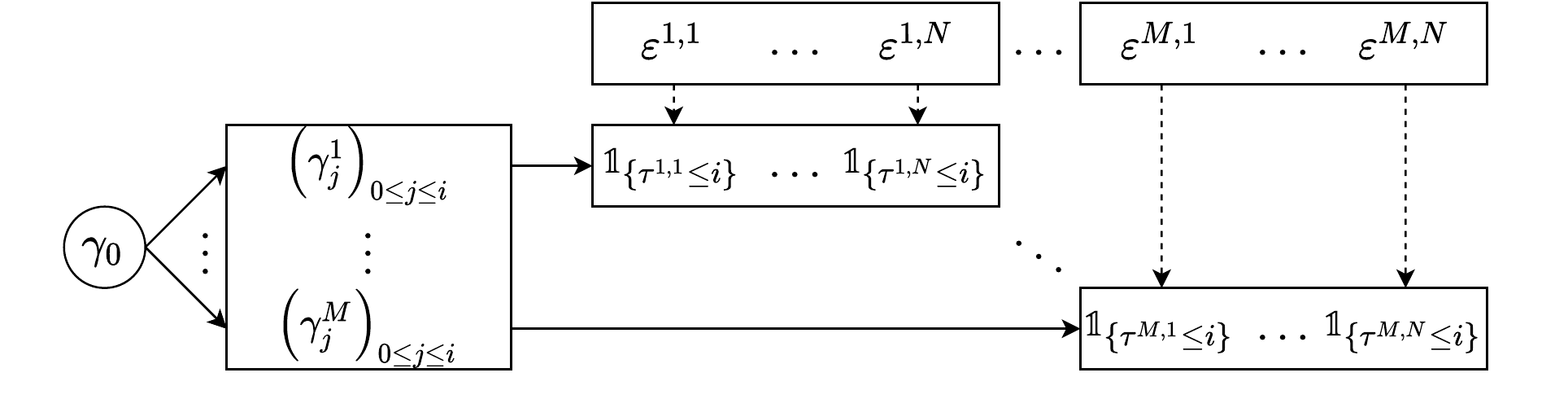}
\caption{Default simulation scheme}
\label{fig:defaults_simulation}
\end{figure}

We then learn the CVA process at different time steps for the whole portfolio at once based on \eqref{eq:CVA_portfolio}, trying different 
combinations of the number of market paths $\my$ and of the \oversimulation factor $\mx$.
Figures~\ref{fig:learning_vs_nmc_rmse_matrix} and \ref{fig:sim_exec_time_cpu_meme_matrices} show the
relative RMSE of the trained neural network against the nested Monte Carlo benchmark\footnote{RMSE restricted to the realizations where the benchmark is non-zero.}, the simulation and training times on the GPU and the host RAM  usage, 
as functions of the number of diffusion paths $M$ and of the \oversimulation factor $N$. 
For the execution times in Figure~\ref{fig:sim_exec_time_cpu_meme_matrices}, the runs were done on a server with an Intel Xeon Gold 6248 CPU and 4 Nvidia Tesla V100 GPUs (out of which we used only one). For performance comparison reasons, we use for all $(M, N)$ configurations the same number of epochs $E=8$ and number of batches $|\mathcal{B}|=32$, which yields a total of $256$ stochastic gradient descent steps during any training task. From Figures~\ref{fig:learning_vs_nmc_rmse_matrix} and \ref{fig:sim_exec_time_cpu_meme_matrices},
we already see some configurations $(\frac{1}{2} M, N)$ being better than $(M, \frac{1}{2} N)$, as they achieve a similar accuracy with  less memory footprint and computational time. For example, $(32768, 1024)$ is better than $(65536, 512)$, given that the former achieves a similar RMSE of 0.07 but is 30\% faster to simulate and price, while also occupying 23\% less CPU memory. In addition, to the credit of the twin Monte Carlo validation procedure of Section \ref{ss:twin}, the comparison between Figures \ref{fig:learning_vs_nmc_rmse_matrix} and \ref{fig:learning_l2err_rmse_matrix} shows that the error estimates provided by the latter are very much in line with the ones\footnote{in spite of the slight variation in the way these errors are computed (cf.~the captions of Figures~\ref{fig:learning_vs_nmc_rmse_matrix} and \ref{fig:learning_l2err_rmse_matrix}).} provided by a much heavier nested Monte Carlo procedure (which would become unfeasible on more complex problems).
\begin{figure}[!htbp]
\centering
\includegraphics[scale=0.4]{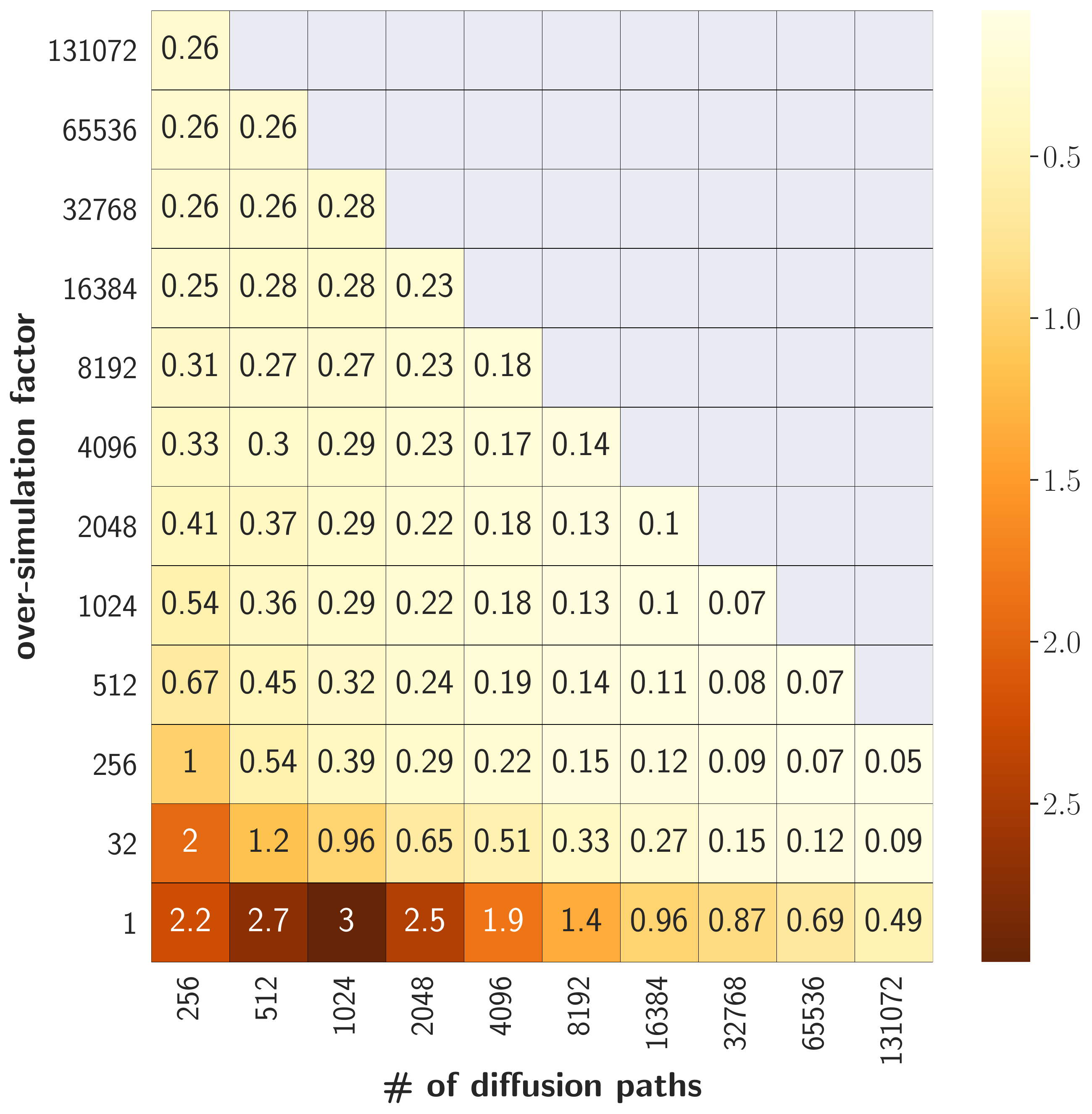}
\caption{Relative RMSE of the prediction against a nested Monte Carlo benchmark at the pricing time $i=5$ years, for different combinations of the number of market paths $\my$ and of the \oversimulation factor $\mx$, when the nested Monte Carlo benchmark is non-zero. The error here is a Monte Carlo estimate of $\sqrt{\mathbb{E}\left[\left(\frac{\mathrm{CVA}_{\text{pred}}-\mathrm{CVA}_{\text{nested}}}{\mathrm{CVA}_{\text{nested}}}\right)^2\right]}$ where $\mathrm{CVA}_{\text{pred}}$ is the CVA estimate predicted by the considered neural network given a state of market and default factors and $\mathrm{CVA}_{\text{nested}}$ is a nested Monte Carlo estimator given the same state.
}
\label{fig:learning_vs_nmc_rmse_matrix}
\end{figure}

\begin{figure}[!htbp]
\centering
\includegraphics[scale=0.35]{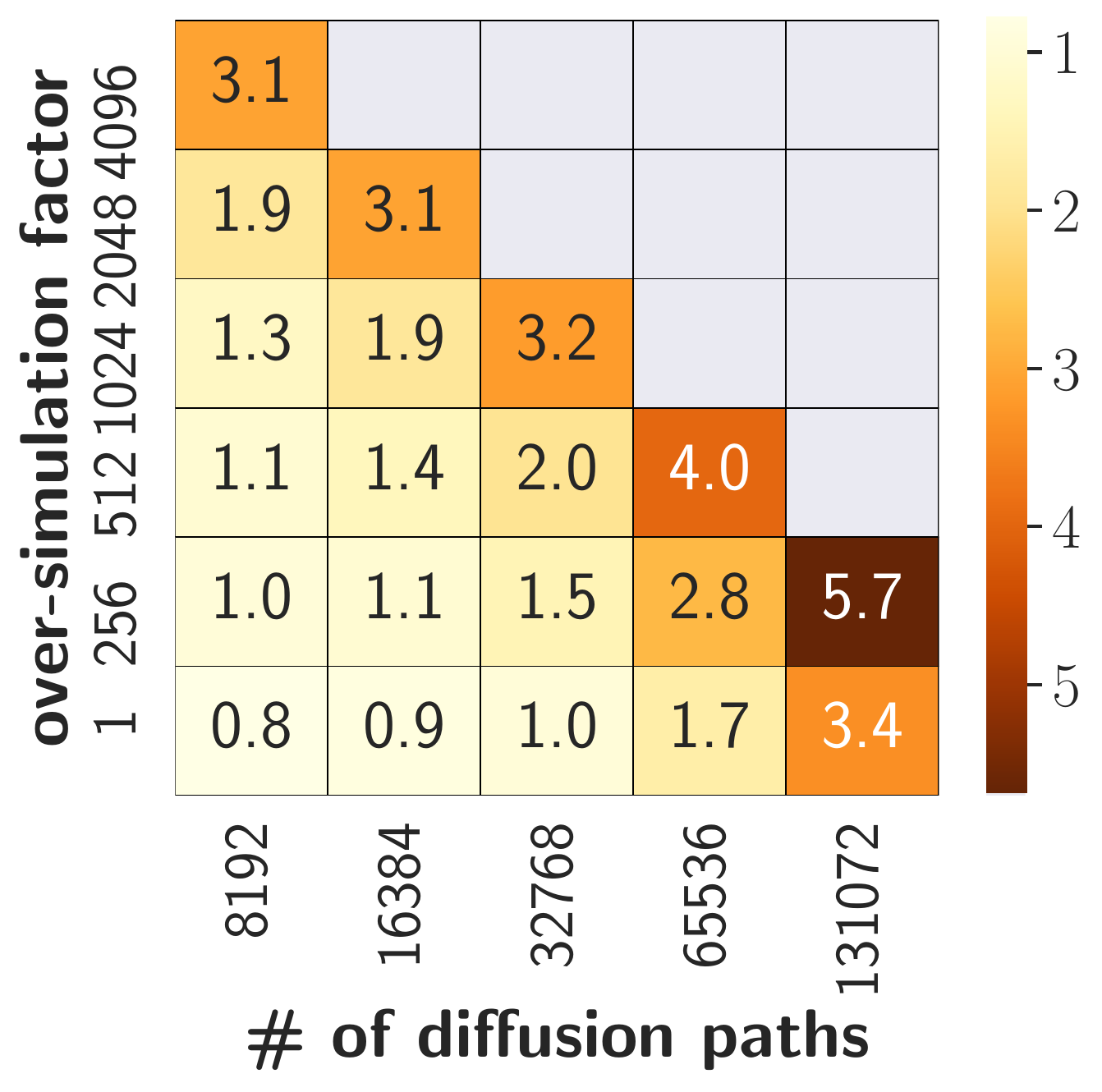}\includegraphics[scale=0.35]{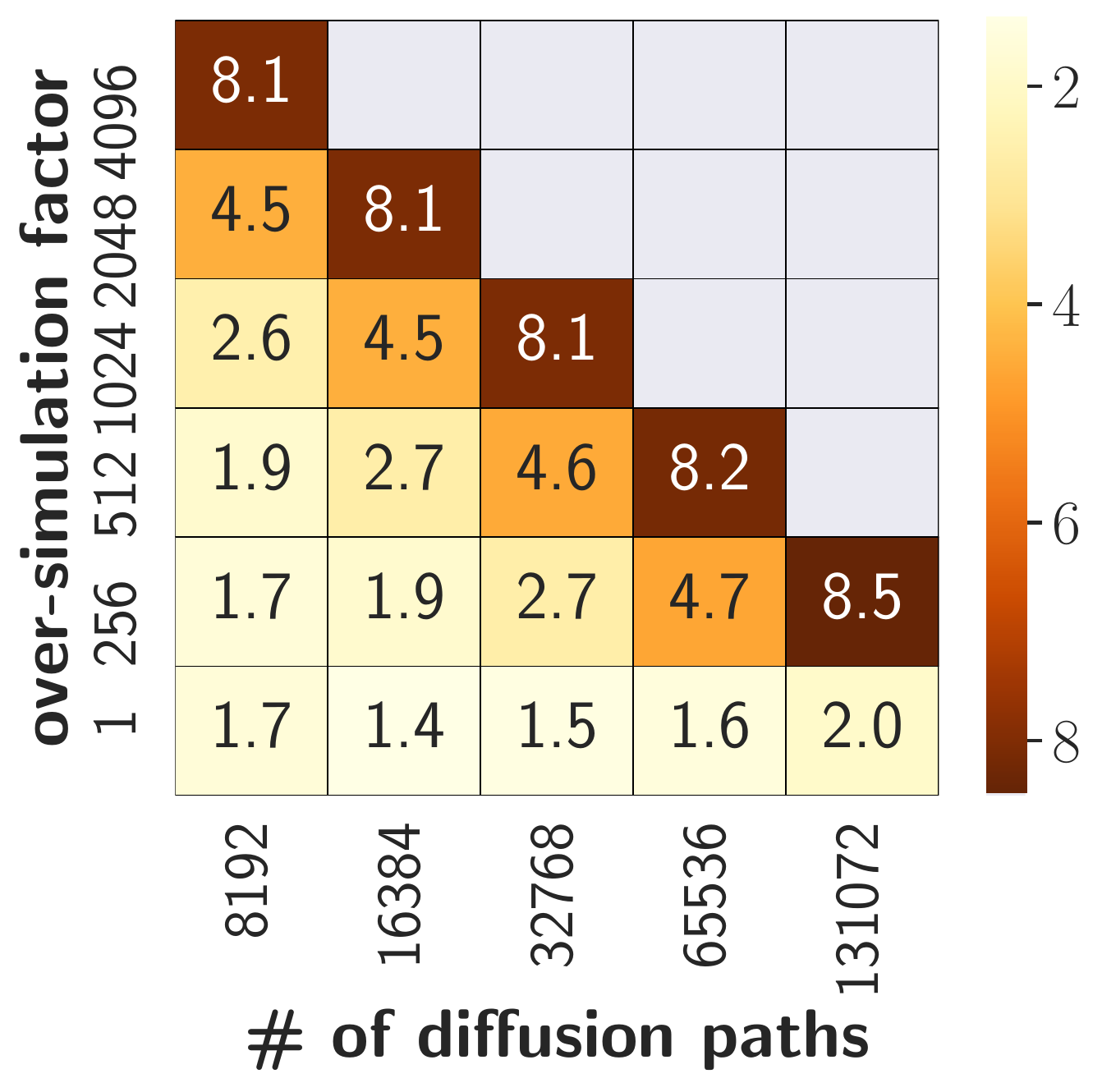}\includegraphics[scale=0.35]{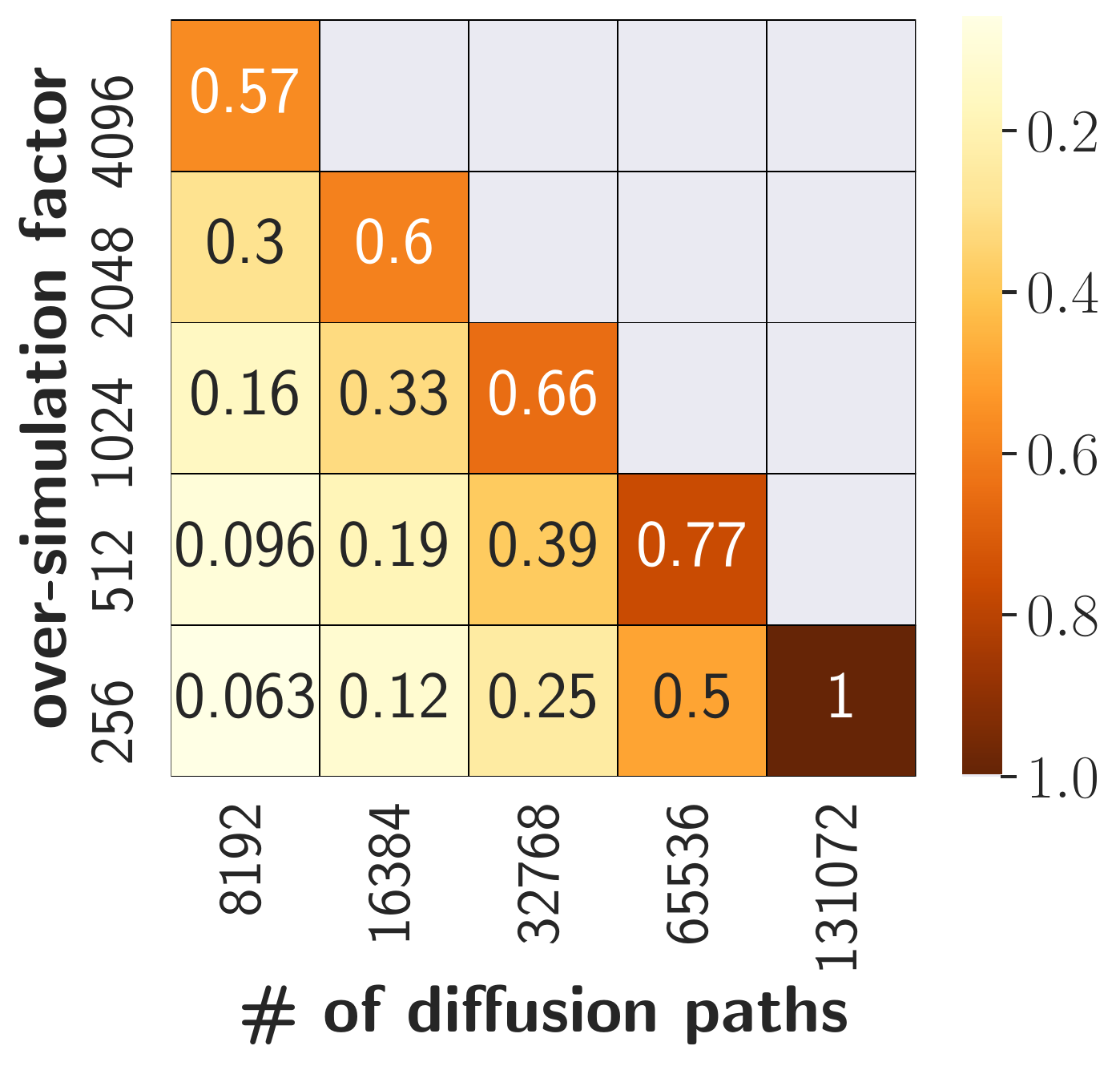}
\caption{Simulation times in seconds (left)  and training times in minutes (center) and RAM usage as a \% of its maximum usage over all the displayed experiments (right), for different combinations of the number of market paths $\my$ and of the \oversimulation factor $\mx$.}
\label{fig:sim_exec_time_cpu_meme_matrices}
\end{figure} 

\begin{figure}[!htbp]
\centering
\includegraphics[scale=0.4]{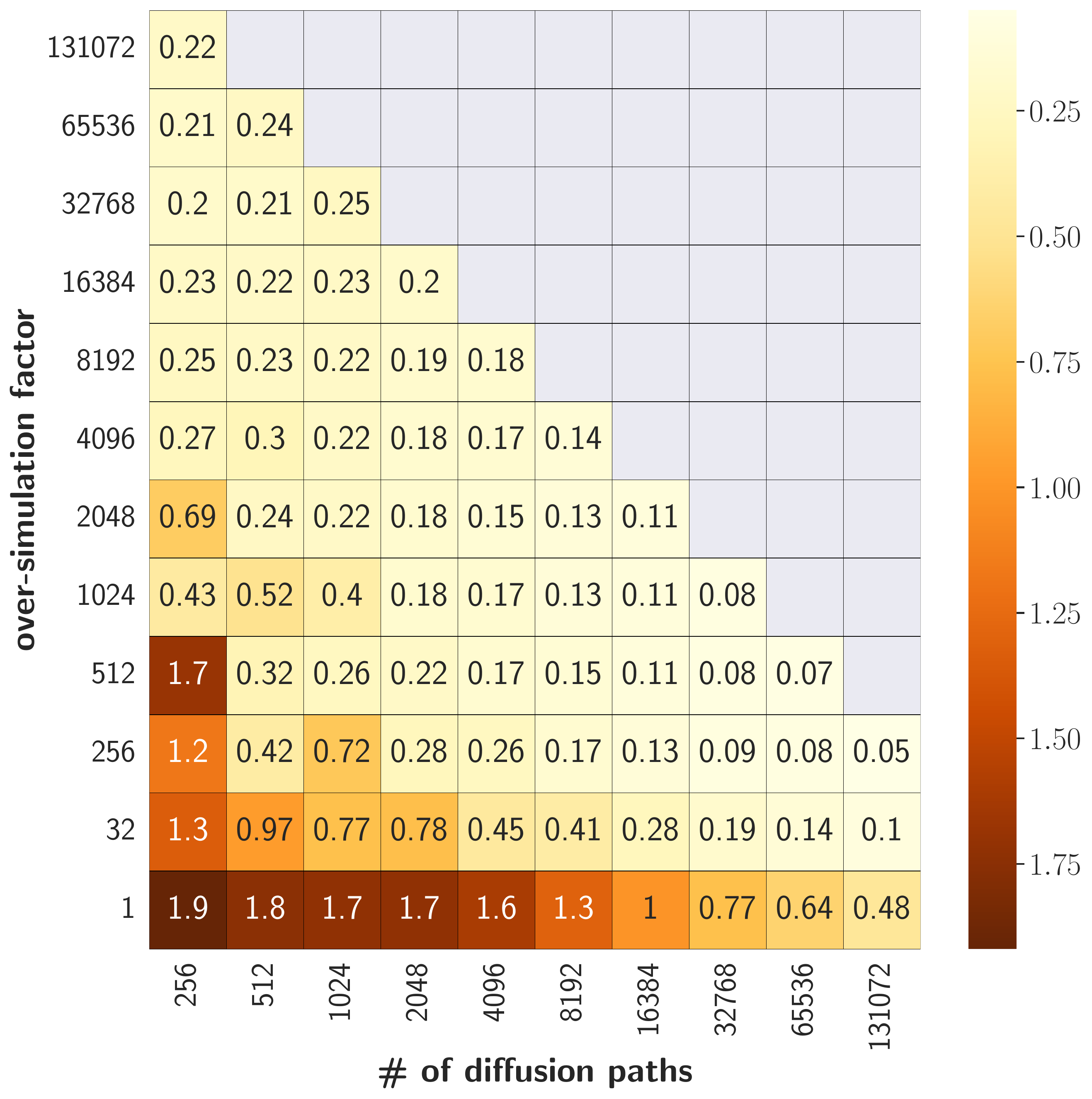}
\caption{Relative RMSE of the prediction against the ground-truth CVA at the pricing time $i=5$ years, computed by the twin simulation approach of Section~\ref{ss:twin} for different combinations of the number of market paths $\my$ and of the \oversimulation factor $\mx$. The error here is a Monte Carlo estimate of $\sqrt{\frac{\mathbb{E}\left[\left(\mathrm{CVA}_{\text{pred}}-\mathrm{CVA}_{\text{exact}}\right)^2\right]}{\mathbb{E}\left[\mathrm{CVA}_{\text{exact}}^2\right]}}$, where $\mathrm{CVA}_{\text{pred}}$ is the CVA estimate predicted by the considered neural network given a state of market and default factors and $\mathrm{CVA}_{\text{exact}}$ is the exact CVA (the value of which does not need to be computed, by virtue of~\eqref{e:val1}-\eqref{e:val2}), given the same state.}
\label{fig:learning_l2err_rmse_matrix}
\end{figure}
The dominance of the impact of the variance of $X$ on that of $\xi$ has been demonstrated in Figure~\ref{fig:ard-boxplot}.
Figure  \ref{f:optom} shows the  $\sqrt{\frac{Q^{\theta}_i
}{R^{\theta}_i}}$ (cf.
 \eqref{e:optom}) obtained in the base case $\mx=1$.
\begin{figure}[h]
  \centering
  \hspace{-1cm}
  \includegraphics[scale=0.55]{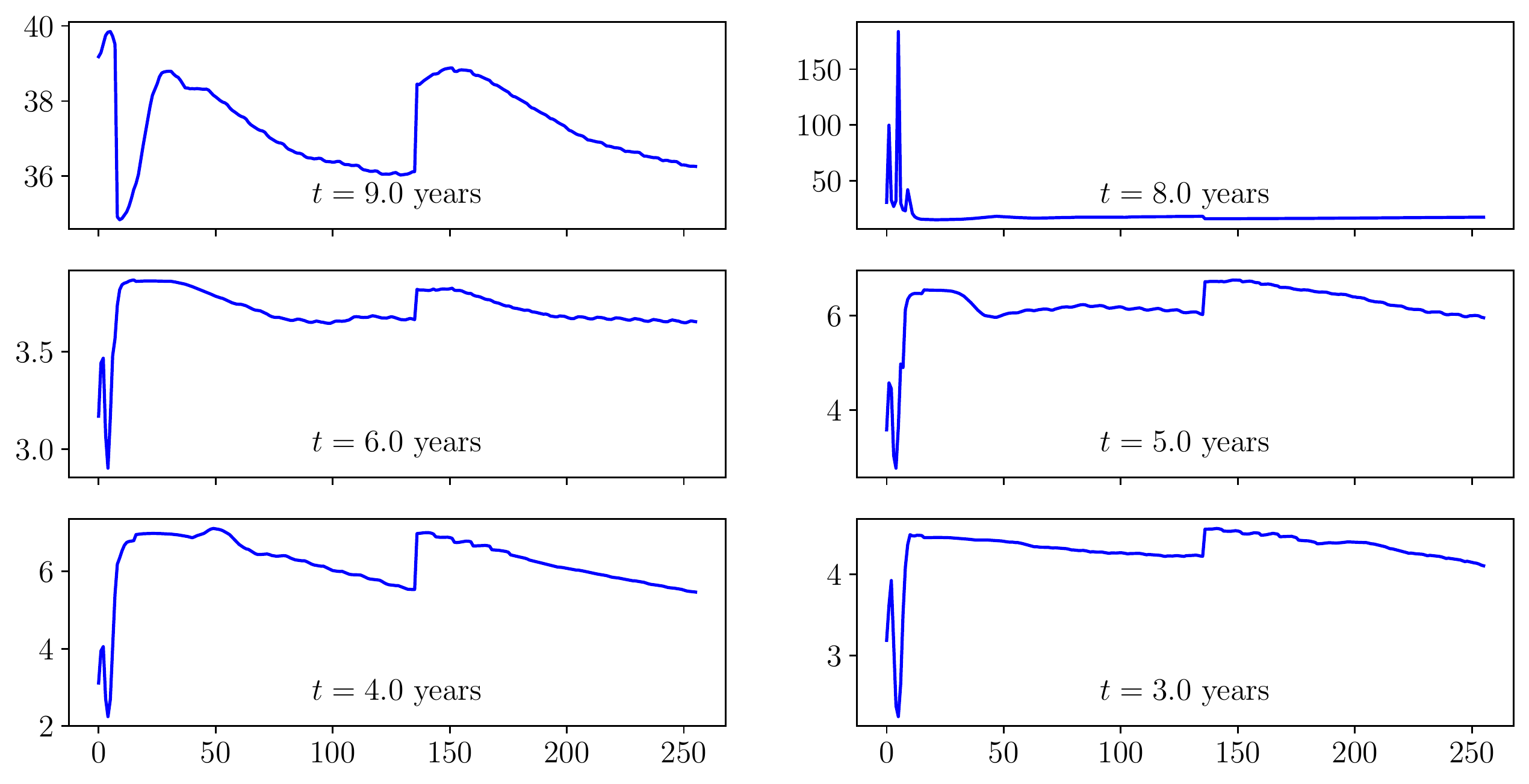}
   \caption{$\sqrt{\frac{Q^{\theta}_i
}{R^{\theta}_i}}$ at different pricing time steps $i$ (panels) and SGD iterations 
($x$ axes).}
  \centering
  \label{f:optom}
\end{figure}  
The values fluctuate quite significantly both in pricing time steps $i$ and SGD iterations. For the purpose of applying Heuristic \ref{heur:MN} (which can be applied for any $N$ but the performance of which needs to be assessed carefully), we retain
a rough average order of magnitude of the order of a few tens.
 To obtain from the $\sqrt{\frac{Q^{\theta}_i
}{R^{\theta}_i}}$ the $\mx^\theta_i$ in \eqref{e:optom}, one needs
  to multiply them by $\sqrt{\kappa}$ (e.g.~if a market simulation is 100 times slower than an ensuing	 default
    simulation, then the factors displayed in Figure \ref{f:optom} must be multiplied by
    10). 
   Solving the equation \eqref{l:ratio} for $P$
on the basis of the
columns $M=65536$ in Figures  \ref{fig:learning_vs_nmc_rmse_matrix}-\ref{fig:sim_exec_time_cpu_meme_matrices} yields $P \approx 497$. So the numbers in Figure~\ref{f:optom} need to be multiplied by $\sqrt{497}\approx 22.3$ to get the optimal $\mx$ as per Heuristic \ref{heur:MN}. 
In view of this, we expect an optimal \oversimulation factor $\mx$ of the order of 
a few hundreds.
  
Path-wise CVA estimators learned for $\mx=1, 32, 64, 128, 512$ are shown in Figure~\ref{fig:table_impact_of_oversim}, which are to be compared to the right plot in Figure~\ref{fig:intensity_vs_indicator} obtained when learning the CVA relying on the intensity-based
formula \eqref{eq:CVAint}. In line with the above expectations, one needs $\mx=512$ in order to have a close enough match between the 1, 2.5, 97.5 and 99-th percentiles of the CVA learned from defaults and those of the nested Monte Carlo estimator (or of the intensity-based CVA learner represented by the right panel in Figure \ref{fig:intensity_vs_indicator}). 
\begin{figure}[h]
\centering
\includegraphics[scale=0.55]{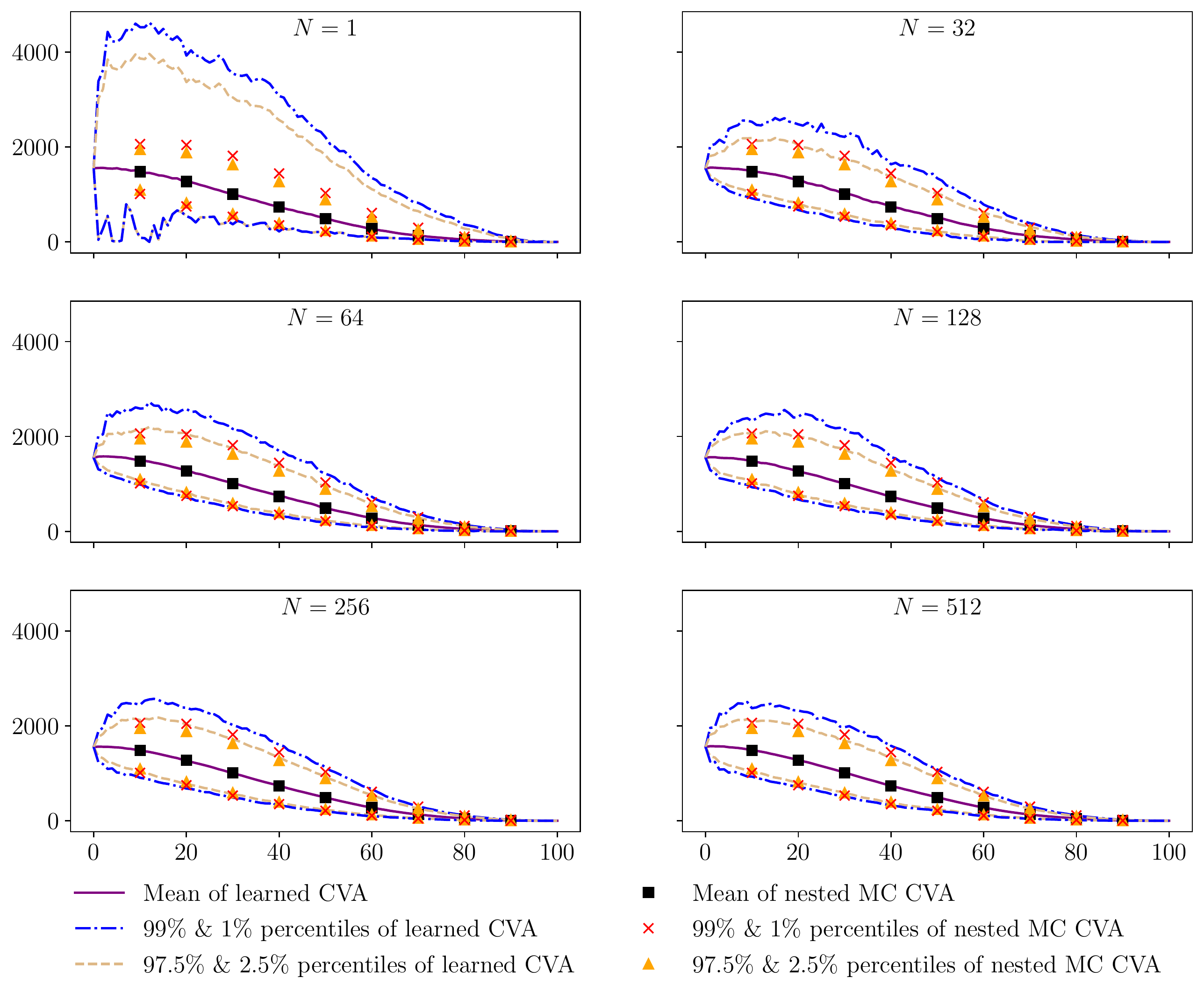}
\caption{
Learned and nested Monte Carlo CVA processes for various \oversimulation factors $N$ ($x$-axis: pricing times, $y$-axis: CVA levels; $\my=16384$). Statistics computed using out-of-sample paths.}
\label{fig:table_impact_of_oversim}
\end{figure}

The above results show that \oversimulation is essential to a defaults-based CVA learner.  The efficiency of Heuristic \ref{heur:MN} is attested numerically on the  double  basis of our nested 
Monte Carlo benchmark (Figure \ref{fig:learning_vs_nmc_rmse_matrix}) and of the twin Monte Carlo validation procedure of Section
\ref{ss:twin}  (Figure \ref{fig:learning_l2err_rmse_matrix}). The fact that Heuristic \ref{heur:MN} 
 already has merit in such case (with not  so stable $\sqrt{\frac{Q^{\theta}_i
}{R^{\theta}_i}}$ ) can be put to the credit of the method.

 \subsection{Conclusion}
The bottom row of Figure \ref{fig:learning_vs_nmc_rmse_matrix}
and the first plot ($\mx=1$) of  Figure \ref{fig:table_impact_of_oversim} illustrate that
 a path-wise CVA cannot be learned based on the hybrid market and defaults formulation \eqref{eq:CVA_portfolio} without \oversimulation: For $\my=16384$ and 131072, the corresponding errors with respect to the benchmark nested Monte Carlo are 96\% and 49\%. However, increasing $\mx$ from 1 (bottom row) to 256
  brings these errors down to 11\% and to 5\%, while for $\my=1024$ and $\mx=512$ the error is 1\%. As visible from Figure \ref{fig:sim_exec_time_cpu_meme_matrices}, the simulation times are only marginally increased when increasing the  \oversimulation factor $\mx$ (while increasing the number of diffusion paths $M$ increases the simulation time approximately proportionally). These results show that the \oversimulation technique is key to the success of a learning approach involving a combination of diffusive and default risk factors.

Even after writing an optimized GPU implementation for the nested Monte Carlo estimator, the latter takes at least 32 minutes on the same hardware as above to compute that estimator for $M=16384$ and $\sqrt{M}=128$ inner paths\footnote{However, when doing the error computations and in all plots, we used $1024$ inner paths to get benchmark CVAs that are sufficiently accurate {\itshape point-wise} and be able to get accurate tail estimates, and nested Monte Carlo simulation thus takes 8 times more computation time.}, compared to approximately 8 minutes in the case of the learning approach with a very high hierarchical simulation factor ($N=2048$). Moreover, going to higher XVA layers such as the FVA and the KVA, a nested Monte Carlo approach would become $\sqrt{\my}$ times slower per each new layer \citep[Section 3.3]{abbas2018xva}, whereas a regression approach would just become slower by a constant each time a new XVA layer is added. In addition, learned XVA metrics can be used in prediction at a very low cost (inference is very fast as it involves no automatic differentiation or stochastic gradient descent), whereas nested Monte Carlo numbers must be recomputed from scratch every time.

In a follow-up paper, the learning-based hierarchical simulation scheme of this paper will be deployed systematically on the whole suite of the XVA  metrics and of the embedded risk measures, resorting to the twin Monte Carlo procedure of Section \ref{ss:twin} for the related validation task (fault of a feasible nested Monte Carlo benchmark beyond the CVA case of this paper).

\appendix

\section{Technical Proofs}
The following proofs use arguments from \citeN{shapiro2014lectures} and extend similar results to the conditionally independent, \oversimulation case. Theorem \ref{bound_finite} extends the finite case in \citeN[Section~5.3.1]{shapiro2014lectures}. The major modifications in the proof are the use of a conditional moment generating function, the establishment of a large deviation upper-bound based on it, and the strict convexity of $\log(\mathbb{E}[\mathcal{M}(\frac{t}{N}, Y)^N])$ with respect to $t$ that becomes more technical in the conditional case.  Then, similar to \citeN[Section~5.3.2]{shapiro2014lectures}, Theorem \ref{bound_infinite} extends these results to the infinite and bounded case, by Lipschitz continuity arguments.  In both cases we rely on the following:


\bl
  \label{largedev_oversim}Let $\varphi : \mathcal{X} \times
  \mathcal{Y} \rightarrow R$ be such that $\varphi (X, Y)$ is integrable, does not degenerate to a constant
  and that, for all $z \in \mathbb{R}$ and $y \in V$, $\mathcal{M} (z, y) \coloneqq \mathbb{E} [\exp (z
  \varphi (X, Y)) |Y = y]$ is well-defined. Then the Fenchel
  conjugate  $I_\mx : a\mapsto
 \sup_{t \in
     \mathbb{R}} \{ ta - \log ( \mathbb{E} [ \mathcal{M} (
     \frac{t}{\mx  }, Y )^\mx   ] ) \}    $ of $t \mapsto \log ( \mathbb{E} [ \mathcal{M} ( \frac{t}{\mx  },
  Y )^\mx   ] )$ is well-defined and
  \begin{eqnarray}\label{e:fench} && \frac{1}{\my  } \log \{ \Q ( \frac{1}{\my\mx}  \sum_{k = 1}^\my  
     \sum_{l = 1}^\mx   \varphi (X^{k, l}, Y^k) \geq a ) \} \leq - I_\mx  
     (a)\;\;\;  \mbox{ for all } a >\mathbb{E} [\varphi (X,
  Y)]  \\&& I_\mx   (a) = \frac{(a -\mathbb{E} [\varphi (X, Y)])^2}{\frac{1}{\mx  }
     \mathbb{E} [\tmop{Var} (\varphi (X, Y) | Y)] + \tmop{Var} (\mathbb{E}
     [\varphi (X, Y) |Y])} + o(| a -\mathbb{E} [\varphi (X, Y)] |^2)\label{e:taylor}
      .\end{eqnarray}
\el

\begin{proof}
  Let $t > 0$. Applying the Markov inequality, we have:
  \begin{eqnarray}
    \Q ( \frac{1}{\my\mx}  \sum_{k = 1}^\my   \sum_{l = 1}^\mx   \varphi
    (X^{k, l}, Y^k) \geq a ) & = & \Q ( \exp (
    \frac{t}{\mx  }  \sum_{k = 1}^\my   \sum_{l = 1}^\mx   \varphi (X^{k, l}, Y^k) )
    \geq \exp (\my  ta) ) \nonumber\\
    & \leq & \exp (- \my  ta) \mathbb{E} [ \exp ( \frac{t}{\mx  }  \sum_{k
    = 1}^\my   \sum_{l = 1}^\mx   \varphi (X^{k, l}, Y^k) ) ] .\label{e:markov}
  \end{eqnarray}
  For every $i \in \{ 1, \ldots, \my   \}$, denote
$Z^i_{\mx} \coloneqq \exp ( \frac{t}{\mx  }  \sum_{j = 1}^\mx   \varphi (X^{i,
     j}, Y^i) )$.
  By using the tower property repeatedly, one can show recursively that for all $i\in\left\{1, \dots, \my\right\}$, denoting  $\mathcal{Z}_i \coloneqq \sigma (Z^1_{\my  , \mx  }, \ldots, Z^{\my   - i}_{\my  , \mx  },
  Y^{\my   - i + 1}, \ldots Y^\my  )$:
  \begin{eqnarray}
    \mathbb{E} [ \exp ( \frac{t}{\mx  }  \sum_{k = 1}^\my   \sum_{l = 1}^\mx  
    \varphi (X^{k, l}, Y^k) ) ] 
    & = & \mathbb{E} [ \mathbb{E} [ \big(\prod_{k = 1}^{\my   - i + 1}
    Z^k_{\mx}\big)  \big(\prod^{i - 1}_{k = 1} \mathcal{M} ( \frac{t}{\mx  }, Y^{\my   - k
    + 1} )^\mx\big)   |\mathcal{Z}_i ] ] \nonumber\\
    & = & \mathbb{E} [ \big(\prod_{k = 1}^{\my   - i} Z^k_{\mx}\big)  \big(\prod^{i -
    1}_{k = 1} \mathcal{M} ( \frac{t}{\mx  }, Y^{\my   - k + 1} )^\mx\big)  
    \mathbb{E} [Z_{\mx}^{\my   - i + 1} |Y^{\my   - i + 1}] ] \nonumber\\
    & = & \mathbb{E} [ \big(\prod_{k = 1}^{\my   - i} Z^k_{\mx}\big)  \big(\prod^i_{k = 1}
    \mathcal{M} ( \frac{t}{\mx  }, Y^{\my   - k + 1} )^\mx\big)   ]
    \nonumber.\\
  \end{eqnarray}
In particular, this identity for $i=M$ yields (recalling $Y$ and the $Y^k$ are i.i.d.)
  \[\mathbb{E} [ \exp ( \frac{t}{\mx  }  \sum_{k = 1}^\my   \sum_{l = 1}^\mx  
    \varphi (X^{k, l}, Y^k) ) ] = ( \mathbb{E} [ \mathcal{M} ( \frac{t}{\mx  }, Y
    )^\mx   ] )^\my.\]
Hence, by \eqref{e:markov},
  \[ \frac{1}{\my  } \log \{ \Q ( \frac{1}{\my\mx}  \sum_{k = 1}^\my  
     \sum_{l = 1}^\mx   \varphi (X^{k, l}, Y^k) \geq a ) \} \leq - ta
     + \log ( \mathbb{E} [ \mathcal{M} ( \frac{t}{\mx  }, Y
     )^\mx   ] ) .\]
  The inequality being true for arbitrary $t > 0$, taking the infimum over $t > 0$ on the
  RHS yields
  \[ \frac{1}{\my  } \log \big( \Q ( \frac{1}{\my\mx}  \sum_{k = 1}^\my  
     \sum_{l = 1}^\mx   \varphi (X^{k, l}, Y^k) \geq a ) \big) \leq -
     \sup_{t > 0} \big( ta - \log ( \mathbb{E} [ \mathcal{M}
     ( \frac{t}{\mx  }, Y )^\mx   ] ) \big) .\]
  In order to establish \eqref{e:fench}, it remains to show that $t \mapsto \log ( \mathbb{E} [ \mathcal{M}
  ( \frac{t}{\mx  }, Y )^\mx   ] )$ is convex and that $I_\mx   (a)
  = \sup_{t > 0} \{ ta - \log ( \mathbb{E} [ \mathcal{M} (
  \frac{t}{\mx  }, Y )^\mx   ] ) \}$. 
  
Define $\Lambda (t)
  \coloneqq \log ( \mathbb{E} [ \mathcal{M} ( \frac{t}{\mx  }, Y
  )^\mx   ] )$. As a moment generating function is infinitely differentiable on its domain of definition,
$\Lambda$ is infinitely differentiable. After computations we get $\Lambda' (0) =\mathbb{E} [\varphi
  (X, Y)]$ and $\Lambda'' (t) = \frac{\det (\mathbb{E} [A])}{\mathbb{E} [\mathbb{E}
     [V|Y]^\mx  ]^2}$
for the $2\times2$ random matrix  
  \[ A = \left[\begin{array}{c|c}
       \frac{1}{\mx  } \mathbb{E} [U^2 V|Y] \mathbb{E} [V|Y]^{\mx   - 1} + ( 1 -
       \frac{1}{\mx  } ) \mathbb{E} [UV|Y]^2 \mathbb{E} [V|Y]^{\mx   - 2} &
       \mathbb{E} [UV|Y] \mathbb{E} [V|Y]^{\mx   - 1}\\
       \hline
       \mathbb{E} [UV|Y] \mathbb{E} [V|Y]^{\mx   - 1} & \mathbb{E} [V|Y]^\mx  
     \end{array}\right] ,\]
where $U \coloneqq \varphi (X, Y)$ and $V \coloneqq \exp ( \frac{t}{\mx  }
  \varphi (X, Y) )$. We have:
  \[ A =\mathbb{E} [V|Y]^{\mx   - 2}  \left[\begin{array}{c|c}
       \frac{1}{\mx  } \mathbb{E} [U^2 V|Y] \mathbb{E} [V|Y] + ( 1 -
       \frac{1}{\mx  } ) \mathbb{E} [UV|Y]^2 & \mathbb{E} [UV|Y] \mathbb{E}
       [V|Y]\\
       \hline
       \mathbb{E} [UV|Y] \mathbb{E} [V|Y] & \mathbb{E} [V|Y]^2
     \end{array}\right] \]
  Let $\alpha, \beta \in \mathbb{R}$ and 
  \[\begin{aligned}
  \Delta_{\alpha, \beta} &\coloneqq \frac{1}{\mathbb{E} [V|Y]^{\mx   - 2}} [\alpha, \beta] A \left[\begin{aligned}\alpha\\\beta\end{aligned}\right] \\
  &=\alpha^2  ( \frac{1}{\mx  } \mathbb{E}
     [U^2 V|Y] \mathbb{E} [V|Y] + ( 1 - \frac{1}{\mx  } ) \mathbb{E}
     [UV|Y]^2 ) + \beta^2 \mathbb{E} [V|Y]^2 + 2 \alpha \beta \mathbb{E}
     [UV|Y] \mathbb{E} [V|Y].
  \end{aligned}\]
  From the Cauchy-Schwarz inequality, we have:
  \begin{equation}
    \label{ineq_cs} \mathbb{E} [U^2 V|Y] \mathbb{E} [V|Y] \geq \mathbb{E}
    [UV|Y]^2 .
  \end{equation}
  We then have:
  \begin{eqnarray}
    \Delta_{\alpha, \beta} & \geq & \alpha^2 \mathbb{E} [UV|Y]^2 + \beta^2
    \mathbb{E} [V|Y]^2 + 2 \alpha \beta \mathbb{E} [UV|Y] \mathbb{E} [V|Y]
    \nonumber\\
    & = & (\alpha \mathbb{E} [UV|Y] + \beta \mathbb{E} [V|Y])^2  \geq  0. \nonumber
  \end{eqnarray}
  Furthermore, we have $\Delta_{\alpha, \beta} > 0$ because $\Delta_{\alpha, \beta} = 0$ would imply equality in
  \eqref{ineq_cs}, which in turn is only attained when $U = 1$ \tmtextit{a.s}., contradicting
  the {non-degeneracy assumption made on $\varphi$}.
Therefore $A$ is a.s. {positive definite}. Hence $\mathbb{E} [A]$ is {positive definite}, i.e.
  $\det (\mathbb{E} [A]) > 0$. In conclusion, $\Lambda$ is {strictly}
  convex.
  
  Let $\psi (t) \coloneqq ta - \Lambda (t)$. For $a >\mathbb{E} [\varphi (X,
  Y)]$, we have
  \[ \psi' (0) = a - \Lambda' (0) = a -\mathbb{E} [\varphi (X, Y)] > 0. \]
  Therefore there exists some $\varepsilon > 0$ such that $\psi' (t) > 0$ for all
  $t \in (0, \varepsilon)$. Hence, for all $t \in (0, \varepsilon)$, we have
  $ta - \Lambda (t) > 0$.
  
  On the other hand, we have:
  \[ \sup_{t < 0} \{ ta - \Lambda (t) \} = \sup_{t < 0} \{ \underbrace{t (a
     -\mathbb{E} [\varphi (X, Y)])}_{< 0} + t\mathbb{E} [\varphi (X, Y)] -
     \Lambda (t) \}. \]
  Using convexity and concavity inequalities, we have
  \begin{eqnarray}
    \Lambda (t) & = & \log ( \mathbb{E} [ \mathcal{M} (
    \frac{t}{\mx  }, Y )^\mx   ] ) \nonumber\\
    & \geq & \mx   \log ( \mathbb{E} [ \mathcal{M} ( \frac{t}{\mx  },
    Y ) ] ) \nonumber\\
    & \geq & t\mathbb{E} [\varphi (X, Y)] .\nonumber
  \end{eqnarray}
  We then obtain that
  \[ \sup_{t < 0} \{ ta - \Lambda (t) \} \leq 0 .\]
  Thus
  \[ \sup_{t > 0} \{ ta - \log ( \mathbb{E} [ \mathcal{M}
     ( \frac{t}{\mx  }, Y )^\mx   ] ) \} = \sup_{t \in
     \mathbb{R}} \{ ta - \log ( \mathbb{E} [ \mathcal{M} (
     \frac{t}{\mx  }, Y )^\mx   ] ) \} = I_\mx   (a), \]
which finishes to prove \eqref{e:fench}.
As the Fenchel conjugate of the twice differentiable and strictly convex function $\Lambda$, $I_\mx  $ is twice differentiable and we have
  \begin{eqnarray}
    I_\mx   (\mathbb{E} [\varphi (X, Y)]) & = & - \Lambda (0) = 0 \nonumber\\
    I_\mx  ' (\mathbb{E} [\varphi (X, Y)]) & = & 0 \nonumber\\
    I_\mx  '' (\mathbb{E} [\varphi (X, Y)]) & = & \frac{1}{\Lambda'' (0)} =
    \frac{1}{\frac{1}{\mx  } \mathbb{E} [\tmop{Var} (\varphi (X, Y) |Y)] +
    \tmop{Var} (\mathbb{E} [\varphi (X, Y) |Y])} .\nonumber
  \end{eqnarray}
  Hence a Taylor expansion around $\mathbb{E} [\varphi (X, Y)]$ gives \eqref{e:taylor}.~\finproof\\
\end{proof}

\brem
The Taylor approximation 
for the Fenchel conjugate in Lemma \ref{largedev_oversim} would become exact
if we added the hypothesis that $\varphi (X, Y)$ is Gaussian conditionally on
$Y$ and that $\mathbb{E} [\varphi (X, Y) |Y]$ is Gaussian. We can
still obtain a similar expression as an exact lower-bound if we just assume
sub-Gaussianity as in Theorems \ref{bound_finite}-\ref{bound_infinite}.
\erem

\subsection{Proof of Theorem \ref{bound_finite}}\label{proof_bound_finite}

Under the assumptions of Theorem \ref{bound_finite}, we have in view of \eqref{e:theevent}:
\begin{equation}
  \label{prob_deltasub1} \Q ( \hat{S}^{\delta}_{\my  , \mx  } (E)
  \not\subset S^{\epsilon} (E) ) \leq \sum_{\theta \in E \setminus
  S^{\epsilon}(E) } \Q ( \bigcap_{\theta' \in E} \{ \hat{G}_{\my  , \mx  }
  (\theta) \leq \hat{G}_{\my  , \mx  } (\theta') + \delta \} )
\end{equation}
Let us consider a minimizer $\theta^\star$ of $G$ over $E$, hence
\begin{equation}
  \label{defmin}\forall \theta \in E \setminus S^{\epsilon}, G   (\theta^\star) < G
   (\theta) - \epsilon . \end{equation}
We then have:
\begin{equation}
  \label{prob_deltasub2} \Q ( \hat{S}^{\delta}_{\my  , \mx  } (E)
  \not\subset S^{\epsilon} (E) ) \leq \sum_{\theta \in E \setminus
  S^{\epsilon}
(E)} \Q (\hat{G}_{\my  , \mx  } (\theta) \leq \hat{G}_{\my  , \mx  }   (\theta^\star) + \delta)
\end{equation}
Define:
\[ \hat{\Gamma}_{\my  , \mx  } (\theta) \coloneqq \frac{1}{\my\mx}  \sum_{k = 1}^\my   \sum_{l =
   1}^\mx   \Gamma ( \theta^\star , \theta, X^{k,l}, Y^k) = \hat{G}_{\my  , \mx  }  
 (\theta^\star) - \hat{G}_{\my  , \mx  } (\theta) \]
Inequality \eqref{prob_deltasub2} can then be rewritten:
\begin{equation}
  \label{prob_deltasub3} \Q ( \hat{S}^{\delta}_{\my  , \mx  }
  \not\subset S^{\epsilon}(E) ) \leq \sum_{\theta \in 
  E
  \setminus
  S^{\epsilon}(E) }\Q (\hat{\Gamma}_{\my  , \mx  } (\theta) \geq - \delta).
\end{equation}

Now observe that
\begin{equation}
\label{eq:mean_leq_minusdelta}
\mathbb{E} [\Gamma ( \theta^\star , \theta, X, Y)] = G  (\theta^\star) - G (\theta)  < - \epsilon < - \delta
\end{equation}
for all $\theta \in E$. Hence, from inequality \eqref{prob_deltasub3} and Lemma
\ref{largedev_oversim}, we obtain
\begin{eqnarray}
  \Q ( \hat{S}^{\delta}_{\my  , \mx  } (E) \not\subset S^{\epsilon} (E)
  ) & \leq & | E | \max_{\theta \in E \setminus S^{\epsilon}(E)} \{
  \Q (\hat{\Gamma}_{\my  , \mx  } (\theta) \geq - \delta) \} \nonumber\\
  & \leq & | E | \exp (- \my   \min_{\theta \in E \setminus S^{\epsilon}(E)} \{
  I^{\theta}_\mx   (- \delta) \})  \label{prob_deltasub4},
\end{eqnarray}
where $I^{\theta}_\mx  $ is the Fenchel conjugate of $\log ( \mathbb{E}
[ \mathcal{M}^{\theta} ( \frac{t}{\mx  }, Y )^\mx   ] )$
and $\mathcal{M}^{\theta} (z, y) \coloneqq \mathbb{E} [\exp (z \Gamma ( 
\theta^\star, \theta, X, Y)) |Y = y]$ for every $\theta \in E$. 
From inequalities
\eqref{subgaussian1} and \eqref{subgaussian2}, we get for every $\theta \in E
\setminus S^{\epsilon}(E)$ and $t \in \mathbb{R}$:
\[ \log ( \mathbb{E} [ \mathcal{M}^{\theta} ( \frac{t}{\mx  }, Y
   )^\mx   ] ) \leq \mathbb{E} [\Gamma ( \theta^\star , \theta, X,
   Y)] t + \frac{1}{2}  ( \frac{b_1^2}{\mx  } + b_2^2 ) t^2 \]
Thus, for every $\theta \in E \setminus S^{\epsilon}(E)$, we have:
\begin{eqnarray}
  I_\mx  ^{\theta} (- \delta) & = & \sup_{t \in \mathbb{R}} \{ t (- \delta -\mathbb{E}
  [\Gamma ( \theta^\star , \theta, X, Y)]) +\mathbb{E} [\Gamma ( \theta^\star , \theta,
  X, Y)] t - \log ( \mathbb{E} [ \mathcal{M}^{\theta} ( \frac{t}{\mx  }, Y
  )^\mx   ] ) \} \nonumber\\
  & \geq & \sup_{t \in \mathbb{R}} \{ t (- \delta -\mathbb{E} [\Gamma ( \theta^\star ,
  \theta, \xi, \Gamma)]) - \frac{1}{2}  ( \frac{b_1^2}{\mx  } + b_2^2 )
  t^2 \} \nonumber\\
  & = & \frac{(- \delta -\mathbb{E} [\Gamma ( \theta^\star , \theta, \xi,
  \Gamma)])^2}{2 ( \frac{b_1^2}{\mx  } + b_2^2 )} \nonumber\\
  & > & \frac{(\epsilon - \delta)^2}{2 ( \frac{b_1^2}{\mx  } + b_2^2
  )}  \label{legendre_lowerbound}
\end{eqnarray}
where the last inequality comes from \eqref{eq:mean_leq_minusdelta}.
This yields the result.

\subsection{Proof of Theorem
\ref{bound_infinite}}\label{proof_bound_infinite}

From the Lipschitz assumption \eqref{lipschitz}, it follows that for all
$\theta, \theta' \in \Theta$:
\[ | \hat{G}_{\my  , \mx  } (\theta) - \hat{G}_{\my  , \mx  } (\theta') | \leq
   \hat{L}_{\my  , \mx  }  \| \theta - \theta' \| \quad \text{a.s.} \]
and
\[ | G (\theta) - G (\theta') | \leq \mathbb{E} [L(X,
   Y)]  \| \theta - \theta' \| \]
where
\[ \hat{L}_{\my  , \mx  } \coloneqq \frac{1}{\my\mx}  \sum_{k = 1}^\my   \sum_{l = 1}^\mx  
   L(X^{k, l}, Y^k) .\]

Let $0 < \delta' < \epsilon'$ and let $\Theta' \coloneqq \{
\theta_1, \ldots, \theta_C \}$ be a minimal $\varrho$-covering of $\Theta$, for a given $\varrho>0$. We then have (cf.~\citeN[Corollary~4.2.13 p.85]{vershynin2018high}):
\[ C \leq 
    ( \frac{2 D}{\varrho} + 1)^d .\]
 Let  $\tilde{\Theta} \coloneqq \Theta' \cup \{\theta^\star\}$,
where  $\theta^\star\in \op{\tmop{Argmin}_{\theta \in \Theta} G (\theta)}$.
 We have $| \tilde{\Theta} | \leq  ( \frac{2 D}{\varrho} + 1
)^d+1 .$ Theorem
\ref{bound_finite} yields:
\[ {\Q ( \hat{S}^{\delta'}_{\my  , \mx  } (\tilde{\Theta})
   \not\subset S^{\epsilon'} (\tilde{\Theta}) ) <  \big( ( \frac{2
   D}{\varrho} +1 )^d+1\big) \exp ( \frac{- \my   (\epsilon' - \delta')^2}{2 (
   \frac{b_1^2}{\mx  } + b_2^2 )} )}. \]
Our next goal is to show the following assertion for suitable choices of
$\delta'$ and $\epsilon'$ and for any $L' > \bar{L}$ (note that  
$\bar{L}<\infty$):	
\begin{equation}
  \label{reduced_to_unreduced} \{\hat{S}^{\delta}_{\my  , \mx  } \not\subset
  S^{\epsilon}\} \cap \{L' \geq \hat{L}_{\my  , \mx  }\} \subset
  \{\hat{S}^{\delta'}_{\my  , \mx  } (\tilde{\Theta}) \not\subset S^{\epsilon'}
  (\tilde{\Theta})\}.
\end{equation}
Let $L' > \bar{L}$ and assume that $\hat{S}^{\delta}_{\my  , \mx  } \not\subset
S^{\epsilon}$ and $L' \geq \hat{L}_{\my  , \mx  }$, and that
$\hat{S}^{\delta'}_{\my  , \mx  } (\tilde{\Theta}) \subset S^{\epsilon'}
(\tilde{\Theta})$. In particular, there exists $\theta \in \Theta$ such that
$\hat{G}_{\my  , \mx  } (\theta) \leq \min_{\Theta} \hat{G}_{\my , \mx  } + \delta$ and $G
(\theta) > \min_{\Theta} G + \epsilon$. Let then $\theta' \subset
\tilde{\Theta}$ such that $\| \theta - \theta' \| < \varrho$. We have:
\[ \hat{G}_{\my  , \mx  } (\theta') \leq \hat{G}_{\my  , \mx  } (\theta) + L' \varrho \leq
   \min_{\Theta} \hat{G}_{\my , \mx  } + \delta + L' \varrho. \]
Thus, if we {choose $\delta' = \delta + L' \varrho$}, then we have
$\hat{G}_{\my  , \mx  } (\theta') \leq \min_{\Theta} \hat{G}_{\my , \mx  } + \delta'$,
and consequently $\theta' \in \hat{S}^{\delta'}_{\my  , \mx  } (\tilde{\Theta})$ (as $\tilde{\Theta}\subset \Theta$).
Hence, by our assumption $\hat{S}^{\delta'}_{\my  , \mx  } (\tilde{\Theta}) \subset S^{\epsilon'}
(\tilde{\Theta})$, we get that $\theta' \in S^{\epsilon'}
(\tilde{\Theta})$. Thus:
\[ G (\theta) \leq G (\theta') + L' \varrho \leq  \min_{\widetilde{\Theta}} G +
   \epsilon' + L' \varrho  = \min_{\Theta} G +
   \epsilon' + L' \varrho ,\]
   as   $ \min_{\widetilde{\Theta}} G
= \min_{ \Theta }G$ (since $\theta^\star\in \tilde{\Theta}$).
Hence, if we {also choose $\epsilon' = \epsilon - L' \varrho$ with
$\varrho$ such that $\varrho < \frac{\epsilon - \delta}{2 L'}$ in order to
ensure that $0 < \delta' < \epsilon'$}, then $G (\theta) \leq
\min_{\Theta} G + \epsilon$,
 which contradicts our assumption 
 and proves 
\eqref{reduced_to_unreduced}. Thus 
\[ \Q ( \{\hat{S}^{\delta}_{\my  , \mx  } \not\subset S^{\epsilon}\} \cap
   \{L' \geq \hat{L}_{\my  , \mx  } \}) \leq \Q (
   \hat{S}^{\delta'}_{\my  , \mx  } (\tilde{\Theta}) \not\subset S^{\epsilon'}
   (\tilde{\Theta}) ) .\]
As a consequence,
\[ \Q ( \hat{S}^{\delta}_{\my  , \mx  } \not\subset S^{\epsilon}
   ) \leq \Q ( \hat{S}^{\delta'}_{\my  , \mx  } (\tilde{\Theta})
   \not\subset S^{\epsilon'} (\tilde{\Theta}) ) +\Q
   (\hat{L}_{\my  , \mx  } > L') .\]
But applying Hoeffding's lemma on the inequalities \eqref{lipschitz_bounded1} and
\eqref{lipschitz_bounded2}, Lemma \ref{largedev_oversim}, and
proceeding similarly as in \eqref{legendre_lowerbound} to establish a lower bound
for the Legendre transform, yields
\beql{e:rhsbis} {\Q (\hat{L}_{\my  , \mx  } > L') \leq \exp (
   \frac{- \my   (L' - \bar{L})^2}{2 ( \frac{\ell_1^2}{\mx  } + \ell_2^2 )}
   )}. \eeql
Thus,
\[ \Q ( \hat{S}^{\delta}_{\my  , \mx  } \not\subset S^{\epsilon}
   ) \leq  \big( ( \frac{2 D}{\varrho}+1 )^d +1\big)\exp ( \frac{- \my  
   (\epsilon' - \delta')^2}{2 ( \frac{b_1^2}{\mx  } + b_2^2 )} )
   + \exp ( \frac{- \my   (L' - \bar{L})^2}{2 ( \frac{\ell_1^2}{\mx  } + \ell_2^2
   )} ). \]
We have $\epsilon' - \delta' = \epsilon - \delta - 2 L' \varrho$. Finally, if we
{choose $\varrho = \frac{\epsilon - \delta}{4 L'}$}, then we get
$\epsilon' - \delta' = \frac{\epsilon - \delta}{2}$ and 
\[ {\Q ( \hat{S}^{\delta}_{\my  , \mx  } \not\subset
   S^{\epsilon} ) \leq \big(  ( \frac{8 L' D}{\epsilon - \delta}+1 )^d +1\big)\exp (
   \frac{- \my   (\epsilon - \delta)^2}{8 ( \frac{b_1^2}{\mx  } + b_2^2 )}
   ) + \exp ( \frac{- \my   (L' - \bar{L})^2}{2 ( \frac{\ell_1^2}{\mx  } +
   \ell_2^2 )} )}. \]
This concludes our proof.

\section{Market and Credit Model in Continuous Time}\label{sec:mktcreditmodelcontinuous}
\def\D{d}
For every economy $e$, the short-rate $r^{\langle e\rangle}$ and the exchange rate $\chi^{\langle e\rangle}$ against the reference currency respectively follow Vasicek and log-normal dynamics  
\beql{e:sde1} 
&\D r^{\langle e\rangle}_t = a^{\langle e\rangle}(b^{\langle e\rangle}-r^{\langle e\rangle}_t) \D t + \sigma^{r, \langle e\rangle} \D \tilde{B}^{r, \langle e\rangle}_t\\
&\D \log\chi^{\langle e\rangle}_t = (r^{\langle 0\rangle}_t-r^{\langle e\rangle}_t-\frac{1}{2}|\sigma^{\chi, \langle e\rangle}|^2) \D t + \sigma^{\chi, \langle e\rangle} \D B^{\chi, \langle e\rangle}_t.\ 
\eeql
For both the bank (``$c=0$'') and every counterparty $c~(\neq 0)$, the process $\gamma^{\langle c\rangle}$ (funding spread for $c=0$ and default intensity for $c \geq 1$) follows CIR dynamics
\beql{e:sde2}
&\D \gamma^{\langle c\rangle}_t = \alpha^{\langle c\rangle}(\delta^{\langle c\rangle}-\gamma^{\langle c\rangle}_t) \D t + \nu^{\langle c\rangle} \sqrt{\gamma^{\langle c\rangle}_t} \D B^{\gamma, \langle c\rangle}_t .
\eeql
In the above, for every $e$, $\tilde{B}^{r,\langle e\rangle}$ is a $\mathbb{Q}^{\langle e\rangle}$ Brownian motion and, for every client $c$ and economy $e$, $B^{\chi,\langle e\rangle}$ and $B^{\gamma,\langle c\rangle}$ are $\mathbb{Q}^{\langle 0\rangle}$ Brownian motions. Here $\mathbb{Q}^{\langle e\rangle}$ is the risk-neutral measure corresponding to the numeraire $\exp(\int_0^{\cdot} r^{\langle e\rangle}_s \D s)$, and $a^{\langle.\rangle}$, $b^{\langle.\rangle}$, $\sigma^{., \langle.\rangle}$, $\alpha^{\langle.\rangle}$, $\delta^{\langle.\rangle}$, $\nu^{\langle.\rangle}$ are model parameters calibrated using liquid market instruments.

In line with the fundamental theorem of asset pricing, for any asset $Z$ priced in a foreign currency 
$e\geq 1$, $\exp(-\int_0^{\cdot} r^{\langle e\rangle}_s \D s) Z $ and $\exp(-\int_0^{\cdot} r^{\langle 0\rangle}_s \D s)
 \chi^{\langle e\rangle}  Z $ are martingales with respect to $\mathbb{Q}^{\langle e\rangle}$ and $\mathbb{Q}^{\langle 0\rangle}$ respectively. In particular,
\[\mathbb{E}^{\mathbb{Q}^{\langle e\rangle}}[e^{-\int_0^t r^{\langle e\rangle}_s \D s} Z_t]=Z_0=\frac{1}{\chi^{\langle e\rangle}_0}\mathbb{E}^{\mathbb{Q}^{\langle 0\rangle}}[e^{-\int_0^t r^{\langle 0\rangle}_s \D s} \chi^{\langle e\rangle}_t Z_t].\]
Thus,
\[
\begin{aligned}
  \frac{\D \mathbb{Q}^{\langle e\rangle}}{\D \mathbb{Q}^{\langle 0\rangle}}_{|t} &=\exp(\int_0^t (r^{\langle e\rangle}_s-r^{\langle 0\rangle}_s) \D s)\frac{\chi^{\langle e\rangle}_t}{\chi^{\langle e\rangle}_0} \\
&=  \exp(-\frac{1}{2}\big(\sigma^{\chi, \langle e\rangle}\big)^2 t + \sigma^{\chi, \langle e\rangle} B^{\chi, \langle e\rangle}_t).
\end{aligned}
\]
Hence, by Girsanov's theorem, if we define $B^{r, \langle e\rangle}_t$ such that:
\[\D \tilde{B}^{r, \langle e\rangle}_t = \D B^{r, \langle e\rangle}_t - \sigma^{\chi, \langle e\rangle}\D\langle B^{r, \langle e\rangle}, B^{\chi, \langle e\rangle}  \rangle_t, \]
then $B^{r, \langle e\rangle}_t$ is a $\mathbb{Q}^{\langle 0\rangle}$ Brownian motion. In particular, assuming $\D\langle B^{r, \langle e\rangle}, B^{\chi, \langle e\rangle}  \rangle_t = \rho^{\langle e\rangle} \D t$, we get the following $\mathbb{Q}^{\langle 0\rangle}$ dynamics for the short-rate of economy $e$:
\[\D r^{\langle e\rangle}_t = (a^{\langle e\rangle}(b^{\langle e\rangle}-r^{\langle e\rangle}_t)-\rho^{\langle e\rangle}\sigma^{\chi, \langle e\rangle}) \D t + \sigma^{r, \langle e\rangle} \D B^{r, \langle e\rangle}_t .\]

For every counterparty $c$, the default time $\tau^{\langle c\rangle}$ can be modeled as a the 
stopping time $\inf\{t>0; \int_0^t \gamma^{\langle c\rangle}_s \D s \geq \epsilon^{\langle c\rangle} \}$, where $\epsilon^{\langle c\rangle}$ is a standard exponential. That is, for every $t \geq 0$,
\begin{equation}\label{e:tau}\ind_{\{\tau^{\langle c\rangle}\le t\}}=1 \Leftrightarrow\tau^{\langle c\rangle} \leq t \Leftrightarrow \int_0^t \gamma^{\langle c\rangle}_s \D s \geq \epsilon^{\langle c\rangle}.\end{equation}

For the instruments, we assume a book comprised of interest rate swaps at par at inception. For each swap, we denote the set of its reset dates by $\mathcal{R}$ and by $t_-$ and $t_+$ the reset dates respectively immediately preceding and following $t$. We assume that successive reset dates are regularly spaced by $\delta$, that the swap is spot starting, i.e.~$0 \in \mathcal{R}$, and that the swap is paying fixed $\delta \Sigma$, where $\Sigma$ is the swap rate, and receiving floating $\frac{1}{\mathrm{ZC}_{t_-}(t)}-1$, where $\mathrm{ZC}_t(t')$ is the price of a zero-coupon bond\footnote{Note that the price of a zero-coupon bond has a closed-form under our affine short-rate model.} at time $t$ with maturity $t'$, at each reset date $t\in\mathcal{R}\setminus\{0\}$. Denoting by $P^{sw}_t$ the price of the swap at time $t$ in units of the underlying currency\footnote{The swap prices are then to be multiplied by the cross-currency exchange rate processes to have all prices in the same reference currency.}, we have for all $t \leq \overline{t} := \max\mathcal{R}$:
\[
P^{sw}_t = \left\{
\begin{aligned}
&\frac{\mathrm{ZC}_t(t_+)}{\mathrm{ZC}_{t_-}(t_+)} - \mathrm{ZC}_t(\overline{t}) - \delta \Sigma \sum_{t' \in \mathcal{R}, t' > t} \mathrm{ZC}_t(t') &\text{ if }& t \notin \mathcal{R}\setminus\{0\}\\
&\frac{1}{\mathrm{ZC}_{t_-}(t)}-\mathrm{ZC}_t(\overline{t})-\delta \Sigma (1+\sum_{t' \in \mathcal{R}, t' > t}\mathrm{ZC}_t(t')) &\text{ if }& t \in \mathcal{R}\setminus\{0\}\\
&1-\mathrm{ZC}_0(\overline{t})-\delta \Sigma \sum_{t'\in\mathcal{R}\setminus\{0\}} \mathrm{ZC}_0(t') &\text{ if }& t=0.
\end{aligned}\right.
\]
\brem\label{rem:pathdep} The path-dependence induced by the  previous reset date can be resorbed by including the short rates of that date among the risk factors $Y$.
\erem

\section{Python/CUDA Optimized Implementation Using GPU\label{ss:pycuda}}

Contrary to most use-cases of machine learning where the final product is the trained model and thus execution time is only critical during inference, in the case of learning from simulated data in pricing applications, the training process itself is part of the final product. Hence particular care is needed when writing the training procedures.

We implemented Algorithm
\ref{alg:backward_algo} using Python programming with the CUDA API (application programming interface). 
Because the considered problem involves high variances (see Section \ref{s:hierarch}) and thus requires a sufficiently large sample size, both training and inference are not easy to achieve in a reasonable execution time. First, we need to leverage the manycore parallel architecture of GPUs that involves streaming multiprocessors, which are used for the simulation, learning and inference phases. All phases are intertwined and performed inline. Hence we need to carefully optimize each part of the algorithm. 

On the simulation side, due to their intrinsically parallel nature, Monte Carlo simulations easily lend themselves to parallelization on GPUs. Nevertheless, various optimizations are needed to achieve a reasonable solution executed within a few seconds 
(cf.~Figure \ref{fig:sim_exec_time_cpu_meme_matrices} in
Section \ref{sec:cva2}).  
We chose to use Python and the CUDA kernels are compiled {\itshape just-in-time} using the module numba, which allows to dynamically generate CUDA kernels at run-time. 

Regarding learning, we opted for PyTorch for its proximity to the CUDA programming model and its {\itshape just-in-time} compiler allowing for static computation graphs and automatic fusion, whenever appropriate, of the kernels associated with the PyTorch operations used by the model.

We used most of the optimization techniques introduced in \citeN{abbas2018xva}, except those related to regressions since these are replaced here by neural networks. 
We also introduced several additional optimizations, the most important one being to judiciously manage the CPU and GPU memories. A naive solution would  involve the CPU/GPU virtual unified memory \citep{CudaCProgGuide} and let the compiler choose. However, this usually results in sub-optimal memory accesses. 
Our choice rather targets an efficient use of the GPU memory space, a reduction of CPU/GPU transfer and an optimized transfer when needed.
These optimizations and implementation choices are developed in the accompanying Github repository\footnote{\url{https://github.com/BouazzaSE/NeuralXVA}, see the coverpage of the paper.}.

Having in mind a portfolio of the order of one million trades spread over maturities ranging over 50 years and involving a few thousands of clients, the computational CPU ressources typically available in banks hardly allow computing (even overnight) a mark-to-market cube with more than $10^4$ paths. 
Switching to GPU ressources (as required anyway if training path-wise XVA metrics is envisioned) could allow computing a mark-to-market cube with $10^5$ to $10^6$ paths in about one hour of computations spread over a few GPUs.

In fact, while we performed our computations using only one GPU, we expect a bank to have access to more than just a single GPU. Monte Carlo simulations and stochastic gradient descent can easily be adapted to multi-GPU setups: see for instance \citep{abbas2014pricing} for a study of the parallelization of a Monte Carlo pricing procedure over multiple GPUs and nodes. As for training, the main parallelization issue is the ability of the optimization algorithm to scale over multiple GPUs or nodes: see in particular \citep{recht2011hogwild} for an asynchronous SGD algorithm which does not require synchronization between the different workers involved. 
Combining these approaches would allow for an implementation that can readily scale to multiple GPUs and nodes, reducing the computation times proportionally to the total number of GPU nodes that are available.

\bibliographystyle{chicago}
\bibliography{\jobname}
\eee